\documentclass[oneside,reqno,fleqn]{amsart}
\usepackage[T1]{fontenc}
\usepackage[latin9]{inputenc}
\usepackage{verbatim}
\usepackage{enumitem}
\usepackage{amstext}
\usepackage{amsthm}
\usepackage{amssymb}
\usepackage{marvosym}

\makeatletter
\numberwithin{equation}{section}
\numberwithin{figure}{section}
\theoremstyle{plain}
\newtheorem{thm}{\protect\theoremname}[section]
  \theoremstyle{plain}
  \newtheorem{fact}[thm]{\protect\factname}
  \theoremstyle{plain}
  \newtheorem{cor}[thm]{\protect\corollaryname}
  \theoremstyle{plain}
  \newtheorem{prop}[thm]{\protect\propositionname}
  \theoremstyle{definition}
  \newtheorem{defn}[thm]{\protect\definitionname}
  \theoremstyle{remark}
  \newtheorem{rem}[thm]{\protect\remarkname}
  \theoremstyle{plain}
  \newtheorem{lem}[thm]{\protect\lemmaname}
  \theoremstyle{plain}
  \newtheorem*{lem*}{\protect\lemmaname}

\usepackage{marvosym}
\usepackage{amsmath}
\usepackage{sherstov}
\usepackage{small-caption}
\usepackage{xcolor}

\usepackage{hyperref}
\hypersetup{
pdfauthor = {V. V. Podolskii, A. A. Sherstov},
plainpages = false,
bookmarks = true,
bookmarksopen = true,
pdftex,
colorlinks = true,
citecolor = teal,
linkcolor = teal,
urlcolor  = teal,
}

\newtheoremstyle{myplain}      {10pt}{10pt}{\itshape}{}{\scshape}{.}{.5em}{}
\newtheoremstyle{mydefinition} {10pt}{10pt}{}{}{\scshape}{.}{.5em}{}
\newtheoremstyle{myremark} {10pt}{10pt}{}{}{\itshape}{.}{.5em}{}

\@namedef{thm}{\@thm{\let \thm@swap \@gobble \th@myplain }{thm}{Theorem}}
\@namedef{lem}{\@thm{\let \thm@swap \@gobble \th@myplain }{thm}{Lemma}}
\@namedef{cor}{\@thm{\let \thm@swap \@gobble \th@myplain }{thm}{Corollary}}
\@namedef{prop}{\@thm{\let \thm@swap \@gobble \th@myplain }{thm}{Proposition}}
\@namedef{claim}{\@thm{\let \thm@swap \@gobble \th@myplain }{thm}{Claim}}
\@namedef{fact}{\@thm{\let \thm@swap \@gobble \th@myplain }{thm}{Fact}}
\@namedef{defn}{\@thm{\let \thm@swap \@gobble \th@mydefinition }{thm}{Definition}}
\@namedef{rem}{\@thm{\let \thm@swap \@gobble \th@mydefinition }{thm}{Remark}}
\@namedef{example}{\@thm{\let \thm@swap \@gobble \th@mydefinition }{thm}{Example}}

\@namedef{thm*}{\@thm {\th@myplain }{}{Theorem}}
\@namedef{lem*}{\@thm {\th@myplain }{}{Lemma}}
\@namedef{cor*}{\@thm {\th@myplain }{}{Corollary}}
\@namedef{prop*}{\@thm {\th@myplain }{}{Proposition}}
\@namedef{claim*}{\@thm {\th@myplain }{}{Claim}}
\@namedef{fact*}{\@thm {\th@myplain }{}{Fact}}
\@namedef{defn*}{\@thm {\th@mydefinition }{}{Definition}}
\@namedef{rem*}{\@thm {\th@mydefinition }{}{Remark}}
\@namedef{example*}{\@thm {\th@mydefinition }{}{Example}}

\makeatletter
\def\@seccntformat#1{%
  \protect\textup{%
    \protect\@secnumfont
    \expandafter\protect\csname format#1\endcsname % <--- added
    \csname the#1\endcsname
    \protect\@secnumpunct
  }%
}

\makeatletter
\newcommand \SparseDotfill {\leavevmode \leaders \hb@xt@ .7em{\hss .\hss }\hfill \kern \z@}
\makeatother	
\makeatletter
\def\@tocline#1#2#3#4#5#6#7{\relax
  \ifnum #1>\c@tocdepth % then omit
  \else
    \par \addpenalty\@secpenalty\addvspace{\ifnum #1=1 2mm \else #2\fi}%
    \begingroup \hyphenpenalty\@M
    \@ifempty{#4}{%
      \@tempdima\csname r@tocindent\number#1\endcsname\relax
    }{%
      \@tempdima#4\relax
    }%
    \parindent\z@ \leftskip#3\relax \advance\leftskip\@tempdima\relax
    \rightskip\@pnumwidth plus4em \parfillskip-\@pnumwidth
          \ifnum #1=1 \bfseries #5\else #5\fi 
   \leavevmode\hskip-\@tempdima
      \ifcase #1
       \or\or \hskip 1em \or \hskip 2em \else \hskip 3em \fi%
#6     \nobreak\relax
{\ifnum #1=1\hfill \else \SparseDotfill\fi}
 \hbox to\@pnumwidth{\@tocpagenum{
    \ifnum #1=1 \bfseries \fi #7}}\par% <---- \dotfill -> \hfill
    \nobreak
    \endgroup
  \fi}
\makeatother
\newcommand{\XOR}{\text{\rm XOR}}
\newcommand{\iu}{\mathrm{i}}

\makeatother

  \providecommand{\corollaryname}{Corollary}
  \providecommand{\definitionname}{Definition}
  \providecommand{\factname}{Fact}
  \providecommand{\lemmaname}{Lemma}
  \providecommand{\propositionname}{Proposition}
  \providecommand{\remarkname}{Remark}
\providecommand{\theoremname}{Theorem}

\begin{document}

\title[Inner Product and Set Disjointness: Beyond Log N Parties]{Inner Product and Set Disjointness: \\Beyond Logarithmically Many
Parties}

\author{Vladimir V. Podolskii$^{*}$ and Alexander A. Sherstov$^{\dagger}$}
\begin{abstract}
A basic goal in complexity theory is to understand the communication
complexity of number-on-the-forehead problems $f\colon(\zoon)^{k}\to\zoo$
with $k\gg\log n$ parties. We study the problems of inner product
and set disjointness and determine their randomized communication
complexity for every $k\geq\log n$, showing in both cases that $\Theta(1+\lceil\log n\rceil/\log\lceil1+k/\log n\rceil)$
bits are necessary and sufficient. In particular, these problems admit
constant-cost protocols if and only if the number of parties is $k\geq n^{\epsilon}$
for some constant $\epsilon>0.$
\end{abstract}

\thanks{$^{*}$ Steklov Mathematical Institute and National Research University
Higher School of Economics, Moscow, Russia. {\large{}\Letter ~}\texttt{podolskii@mi.ras.ru}\\
$\phantom{a}$\quad{}$^{\dagger}$ Computer Science Department, UCLA,
Los Angeles, CA~90095. {\large{}\Letter ~}\texttt{sherstov@cs.ucla.edu
}Supported by NSF CAREER award CCF-1149018 and an Alfred P. Sloan
Foundation Research Fellowship.}

\maketitle
\thispagestyle{empty}

\section{Introduction}

The \emph{number-on-the-forehead} model, due to Chandra et al.~\cite{cfl83multiparty},
is the standard model of multiparty communication. The model features
$k$ collaborative players and a Boolean function $F\colon X_{1}\times X_{2}\times\cdots\times X_{k}\to\zoo$
with $k$ arguments. An input $(x_{1},x_{2},\dots,x_{k})$ is distributed
among the $k$ players with overlap, by giving the $i$th player the
arguments $x_{1},\dots,x_{i-1},x_{i+1},\dots,x_{k}$ but not $x_{i}$.
This arrangement can be visualized as having the $k$ players seated
in a circle with $x_{i}$ written on the $i$th player's forehead,
whence the name of the model. The players communicate according to
a protocol agreed upon in advance. The communication occurs in the
form of broadcasts, with a message sent by any given player instantly
reaching everyone else. The players' objective is to compute $F$
on any given input with minimal communication. We are specifically
interested in \emph{randomized} protocols, where the players have
an unbounded supply of shared random bits. The \emph{cost} of a protocol
is the total bit length of all the messages broadcast in a worst-case
execution. The\emph{ $\epsilon$-error randomized communication complexity
$R_{\epsilon}(F)$} is the least cost of a randomized protocol that
computes $F$ with probability of error at most $\epsilon$ on every
input.

Number-on-the-forehead communication complexity is a natural subject
of study in its own right, in addition to its applications to circuit
complexity, pseudorandomness, and proof complexity~\cite{bns92,yao90acc,hastad-goldman91multiparty,razborov-wigderson93multiparty,bps07lovasz-schrijver}.
Number-on-the-forehead is the most studied model in the area because
any other way of assigning arguments to players results in a less
powerful formalism\textemdash provided of course that one does not
assign all the arguments to some player, in which case there is never
a need to communicate. The generous overlap in the players' inputs
makes proving lower bounds in the number-on-the-forehead model difficult.
The strongest lower bound for an explicit communication problem $F\colon(\zoon)^{k}\to\zoo$
is currently $\Omega(n/2^{k}),$ obtained by Babai et al.~\cite{bns92}
almost thirty years ago. This lower bound becomes trivial at $k=\log n,$
and it is a longstanding open problem to overcome this logarithmic
barrier and prove strong lower bounds for an explicit function with
$k\gg\log n$. As one would expect, the existence of such functions
is straightforward to prove using a counting argument~\cite{bns92,ford-thesis}.
In particular, it is known~\cite[Sec.~9.2]{ford-thesis} that for
all $n$ and $k,$ a uniformly random function $F\colon(\zoon)^{k}\to\zoo$
almost surely has randomized communication complexity
\begin{equation}
R_{1/3}(F)\geq n-5,\label{eq:random-fn-cc}
\end{equation}
which essentially meets the trivial upper bound $R_{1/3}(F)\leq n+1.$

The two most studied problems in communication complexity theory are
(\emph{generalized}) \emph{inner product} and \emph{set disjointness}.
In the $k$-party versions of these problems, the inputs are subsets
$S_{1},S_{2},\ldots,S_{k}\subseteq\{1,2,\ldots,n\}.$ As usual, the
$i$th player knows $S_{1},\ldots,S_{i-1},S_{i+1},\ldots,S_{k}$ but
not $S_{i}$. In the inner product problem, the objective is to determine
whether $|\bigcap S_{i}|$ is odd. In set disjointness, the objective
is to determine whether $\bigcap S_{i}=\varnothing$. In Boolean form,
these two functions are given by the formulas
\begin{align*}
\GIP_{n,k}(X)= & \bigoplus_{i=1}^{n}\bigwedge_{j=1}^{k}X_{i,j},\\
\DISJ_{n,k}(X)= & \bigwedge_{i=1}^{n}\bigvee_{j=1}^{k}\overline{X}_{i,j},
\end{align*}
respectively, where the input is an $n\times k$ Boolean matrix $X\in\zoo^{n\times k}$
whose columns are the characteristic vectors of the input sets. In
the setting of two players, the communication complexity is well-known
to be $\Theta(n)$ for both inner product~\cite{chor-goldreich88ip}
and set disjointness~\cite{KS92disj,razborov90disj,baryossef04info-complexity}.
A moment's thought reveals that the $k$-party communication complexity
of these problems is monotonically nonincreasing in $k,$ and determining
this dependence has been the subject of extensive research in the
area~\cite{bns92,grolmusz94multi-upper,tesson03thesis,beame06corruption,lee-shraibman08disjointness,chatt-ada08disjointness,beame-huyn-ngoc09multiparty-focs,sherstov12mdisj,sherstov13directional}.
On the upper bounds side, Grolmusz~\cite{grolmusz94multi-upper}
proved that $k$-party inner product has communication complexity
$O(k\lceil n/2^{k}\rceil),$ which easily carries over to $k$-party
set disjointness. The best lower bounds to date are $\Omega(n/4^{k})$
for inner product, due to Babai et al.~\cite{bns92}; and $\Omega(n/4^{k})^{1/4}$
and $\Omega(\sqrt{n}/2^{k}k)$ for set disjointness, due to Sherstov~\cite{sherstov12mdisj,sherstov13directional}.

\subsection{Our results}

Our work began with a basic question: how many players $k$ does it
take to compute inner product and set disjointness with \emph{constant}
communication? As discussed above, the best bounds on the communication
complexity of these functions for large $k$ prior to this paper were
$\Omega(1)$ and $O(\log n)$. We close this logarithmic gap, determining
the communication complexity up to a multiplicative constant for every
$k\geq\log n$.

\bigskip{}
\begin{thm}[Main result]
\label{thm:MAIN-gip-disj}For any $k\geq\log n,$ inner product and
set disjointness have randomized communication complexity
\begin{align*}
R_{1/3}(\GIP_{n,k}) & =\Theta\left(\frac{\log n}{\log\left\lceil 1+\frac{k}{\log n}\right\rceil }+1\right),\\
R_{1/3}(\DISJ_{n,k}) & =\Theta\left(\frac{\log n}{\log\left\lceil 1+\frac{k}{\log n}\right\rceil }+1\right).
\end{align*}
\end{thm}

\noindent To our knowledge, Theorem~\ref{thm:MAIN-gip-disj} is the
first nontrivial (i.e., superconstant) lower bound for any explicit
communication problem with $k\geq\log n$ players. In particular,
inner product and set disjointness have communication protocols with
constant cost if and only if the number of players is $k\geq n^{\epsilon}$
for some constant $\epsilon>0.$ It is noteworthy that we prove the
upper bounds in Theorem~\ref{thm:MAIN-gip-disj} using \emph{simultaneous}
protocols, where the players do not interact. In more detail, each
player in a simultaneous protocol broadcasts all his messages at once
and without regard to the messages from the other players. The output
of a simultaneous protocol is fully determined by the shared randomness
and the concatenation of the messages broadcast by all the players.
The cost of a simultaneous protocol is defined in the usual way, as
the total number of bits broadcast by all the players in a worst-case
execution. Theorem~\ref{thm:MAIN-gip-disj} shows that as far as
inner product and set disjointness are concerned, simultaneous protocols
are asymptotically as powerful as general ones.

A natural next step is to construct a problem $F_{n,k}\colon(\zoon)^{k}\to\zoo$
whose communication complexity remains nontrivial for all $k.$ Its
\emph{existence} follows from the lower bound~(\ref{eq:random-fn-cc})
on the communication complexity of random functions. In the theorem
below, we give an explicit function with communication complexity
at least $c\log n$ for some absolute constant $c>0$ and all $n$
and $k.$ We remind the reader that $\MOD_{m}$ stands for the Boolean
function that evaluates to true if and only if the sum of its arguments
is a multiple of $m.$ 
\begin{thm}
\label{thm:MAIN-mod3-mod2}Define $F_{n,k}\colon(\zoon)^{k}\to\zoo$
by
\[
F_{n,k}(X)=\MOD_{3}\left(\bigoplus_{j=1}^{k}X_{1,j},\ldots,\bigoplus_{j=1}^{k}X_{n,j}\right).
\]
Then 
\begin{align*}
R_{1/3}(F_{n,k}) & \geq\frac{1}{3}\log n-\frac{1}{3}.
\end{align*}
\end{thm}

\noindent As with inner product and set disjointness, we show that
the lower bound of Theorem~\ref{thm:MAIN-mod3-mod2} is asymptotically
tight for all $k\geq\log n$.

\subsection{Our techniques}

The upper bounds in Theorem~\ref{thm:MAIN-gip-disj} are based on
Grolmusz's deterministic protocol for multiparty inner product~\cite{grolmusz94multi-upper},
which we are able to speed up using public randomness. The lower bounds
in Theorems~\ref{thm:MAIN-gip-disj} and~\ref{thm:MAIN-mod3-mod2}
are more subtle. First of all, it may be surprising that we are able
to prove any lower bounds at all for $k\gg\log n$ players since all
known techniques for explicit functions stop working at $k=\log n$
players. The key is to realize that we only need to rule out communication
protocols with cost $O(\log n)$, and in any given execution of such
a protocol all but $O(\log n)$ players remain silent! This makes
it possible to reduce the analysis to the setting of $k\leq\log n$
players, where strong lower bounds are known. This reduction involves
constructing an input distribution such that the portion of the input
seen by any small set of players does not significantly help with
computing the output. Our communication lower bounds use the \emph{discrepancy
method}, which we adapt here to reflect the number of active players.

The remainder of this paper is organized as follows. Section~\ref{sec:prelim}
gives a thorough review of the technical preliminaries. Our results
for inner product and set disjointness are presented in Sections~\ref{sec:gip}
and~\ref{sec:disj}, respectively. Section~\ref{sec:indep-of-k}
concludes the paper with a proof of Theorem~\ref{thm:MAIN-mod3-mod2}
along with a matching upper bound.

\section{Preliminaries\label{sec:prelim}}

\subsection{General}

We use lowercase letters for vectors and strings, and uppercase letters
for matrices. The empty string is denoted~$\varepsilon.$ For a bit
string $x\in\zoon,$ we let $|x|=x_{1}+x_{2}+\cdots+x_{n}$ denote
the Hamming weight of $x.$  We let the bar operator $\overline{\phantom{a}}$
denote either complex conjugation or set complementation, depending
on the nature of the argument. For convenience, we adopt the convention
that $0/0=0.$ The notation $\log x$ refers to the logarithm of $x$
to base $2.$

We will view Boolean functions as mappings $f\colon X\to\zoo$ for
a finite set $X$, typically $X=\zoon.$ A \emph{partial function}
$f$ on a set $X$ is a function whose domain of definition, denoted
$\dom f,$ is a proper subset of $X.$ For (possibly partial) Boolean
functions $f$ and $g$ on $\zoon$ and $X,$ respectively, the symbol
$f\circ g$ refers to the (possibly partial) Boolean function on $X^{n}$
given by $(f\circ g)(x_{1},x_{2},\dots,x_{n})=f(g(x_{1}),g(x_{2}),\dots,g(x_{n})).$
Clearly, the domain of $f\circ g$ is the set of all $(x_{1},\dots,x_{n})\in(\dom g)^{n}$
for which $(g(x_{1}),g(x_{2}),\dots,g(x_{n}))\in\dom f.$ As usual,
for (possibly partial) Boolean functions $f$ and $g$, the symbol
$f\oplus g$ refers to the (possibly partial) Boolean function given
by $(f\oplus g)(x,y)=f(x)\oplus g(y).$ Observe that in this notation,
$f\oplus f$ and $f$ are completely different functions. We abbreviate
$f^{\oplus n}=f\oplus f\oplus\cdots\oplus f$ $(n$ times). The familiar
functions $\AND_{n},$ $\OR_{n},$ and $\XOR_{n}$ on the Boolean
hypercube $\zoon$ are given by $\AND_{n}(x)=\bigwedge_{i=1}^{n}x_{i},$
\, $\OR_{n}(x)=\bigvee_{i=1}^{n}x_{i},$ and $\XOR_{n}(x)=\bigoplus_{i=1}^{n}x_{i}.$
We let $\MOD_{3}\colon\zoo^{*}\to\zoo$ be the Boolean function given
by $\MOD_{3}(x)=1\;\Leftrightarrow\;|x|\equiv0\pmod3.$ Finally, we
define a partial Boolean function $\UAND_{n}$ on $\zoon$ as the
restriction of $\AND_{n}$ to $\{x:|x|\geq n-1\}.$ In other words,
\[
\UAND_{n}(x)=\begin{cases}
x_{1}\wedge x_{2}\wedge\cdots\wedge x_{n} & \text{if }|x|\geq n-1,\\
\text{undefined} & \text{otherwise.}
\end{cases}
\]

We let $X^{n\times k}$ denote the family of $n\times k$ matrices
with entries in $X,$ the most common cases being those of real matrices
$(X=\Re)$ and Boolean matrices $(X=\zoo).$ For a matrix $M\in\Re^{n\times m}$
and a set $S\subseteq\{1,2,\dots,n\},$ we let $M|_{S}$ denote the
submatrix of $M$ obtained by keeping the rows with index in $S.$
More generally, for sets $S\subseteq\{1,2,\dots,n\}$ and $T\subseteq\{1,2,\dots,m\},$
we let $M|_{S,T}$ denote the $|S|\times|T|$ submatrix of $M$ obtained
by keeping the rows with index in $S$ and columns with index in $T.$
We adopt the standard convention that the ordering of the rows (and
columns) in a submatrix is inherited from the containing matrix. 

For nonnegative integers $n$ and $k,$ we define
\[
\binom{n}{\mathord\leq k}:=\binom{n}{0}+\binom{n}{1}+\cdots+\binom{n}{k}=\sum_{i=0}^{\min\{k,n\}}\binom{n}{i}.
\]
The following bounds are well-known~\cite[Proposition~1.4]{jukna01extremal}:
\begin{align}
\left(\frac{n}{k}\right)^{k}\leq\binom{n}{\mathord\leq k} & \leq\left(\frac{\e n}{k}\right)^{k} &  & (1\leq k\leq n).\label{eq:binom-sum-bound}
\end{align}

\subsection{Analytic preliminaries}

For a finite set $X,$ we let $\Re^{X}$ denote the linear space of
real functions $X\to\Re$. This space is equipped with the usual norms
and inner product: 
\begin{align*}
\|\phi\|_{\infty} & =\max_{x\in X}\,|\phi(x)| &  & \qquad(\phi\in\Re^{X}),\\
\|\phi\|_{1} & =\sum_{x\in X}\,|\phi(x)| &  & \qquad(\phi\in\Re^{X}),\\
\langle\phi,\psi\rangle & =\sum_{x\in X}\phi(x)\psi(x) &  & \qquad(\phi,\psi\in\Re^{X}).
\end{align*}
The support of $\phi\in\Re^{X}$ is the subset $\Supp\phi=\{x\in X:\phi(x)\ne0\}.$
The pointwise (Hadamard) product of $\phi,\psi\in\Re^{X}$ is denoted
$\phi\cdot\psi\in\Re^{X}$ and given by $(\phi\cdot\psi)(x)=\phi(x)\psi(x).$
The tensor product of $\phi\in\Re^{X}$ and $\psi\in\Re^{Y}$ is the
function $\phi\otimes\psi\in\Re^{X\times Y}$ given by $(\phi\otimes\psi)(x,y)=\phi(x)\psi(y).$
The tensor product $\phi\otimes\phi\otimes\cdots\otimes\phi$ ($n$
times) is abbreviated $\phi^{\otimes n}.$  Tensor product notation
generalizes to partial functions in the natural way: if $\phi$ and
$\psi$ are partial real functions on $X$ and $Y,$ respectively,
then $\phi\otimes\psi$ is a partial function on $X\times Y$ with
domain $\dom\phi\times\dom\psi$ and is given by $(\phi\otimes\psi)(x,y)=\phi(x)\psi(y)$
on that domain. Similarly, $\phi^{\otimes n}$ is a partial function
on $X^{n}$ with domain $(\dom\phi)^{n}.$

We now recall the Fourier transform on $\zoon.$ For a subset $S\subseteq\oneton,$
define $\chi_{S}\colon\zoon\to\{-1,+1\}$ by $\chi_{S}(x)=\prod_{i\in S}(-1)^{x_{i}}.$
Then every function $\phi\colon\zoon\to\Re$ has a unique representation
of the form $\phi=\sum_{S}\hat{\phi}(S)\,\chi_{S},$ where $\hat{\phi}(S)=\Exp_{x\in\zoon}\phi(x)\chi_{S}(x).$
The reals $\hat{\phi}(S)$ are the \emph{Fourier coefficients} of\emph{
$\phi$}, and the mapping $\phi\mapsto\hat{\phi}$ is the \emph{Fourier
transform} of $\phi$. 

\subsection{Probability}

We view probability distributions first and foremost as real functions
and use the notational shorthands above. In particular, we write $\supp\mu$
to refer to the support of the probability distribution $\mu$, and
$\mu\otimes\lambda$ to refer to the Cartesian product of the distributions
$\mu$ and $\lambda$. The notation $X\sim\mu$ means that the random
variable $X$ is distributed according to $\mu.$ We let $B(n,p)$
denote the binomial distribution with $n$ trials and success probability
$p.$ 
\begin{fact}
\label{fact:binomial-expectations}For any integer $n\geq1,$
\begin{align}
\Exp_{s\sim B(n-1,p)}\;\frac{1}{\sqrt{n-s}} & \leq\frac{1}{\sqrt{(1-p)n}},\label{eq:exp-bin-n-s}\\
\Exp_{s\sim B(n-1,q)}\;\frac{1}{\sqrt{s+1}} & \leq\frac{1}{\sqrt{qn}},\label{eq:exp-bin-s+1}\\
\Exp_{s\sim B(n,p)}\;|s-pn| & \leq\sqrt{p(1-p)n}.\label{eq:exp-bin-L1}
\end{align}
 
\end{fact}

\begin{proof}
For (\ref{eq:exp-bin-n-s}), we have
\begin{align*}
\left(\Exp_{s\sim B(n-1,p)}\;\frac{1}{\sqrt{n-s}}\right)^{2} & \leq\Exp_{s\sim B(n-1,p)}\;\frac{1}{n-s}\\
 & =\sum_{s=0}^{n-1}\binom{n-1}{s}\frac{p^{s}(1-p)^{n-1-s}}{n-s}\\
 & =\frac{1}{(1-p)n}\sum_{s=0}^{n-1}\binom{n}{s}p^{s}(1-p)^{n-s}\\
 & =\frac{1-p^{n}}{(1-p)n},
\end{align*}
where the first step follows from the Cauchy-Schwarz inequality, and
the last step uses the binomial theorem. The bound (\ref{eq:exp-bin-s+1})
follows from (\ref{eq:exp-bin-n-s}) since the distribution of $s+1$
for $B(n-1,q)$ is the same as the distribution of $n-s$ for $s\sim B(n-1,1-q)$.
For (\ref{eq:exp-bin-L1}), 
\begin{align*}
\left(\Exp_{s\sim B(n,p)}\;|s-pn|\right)^{2} & \leq\Exp_{s\sim B(n,p)}[(s-pn)^{2}]\\
 & =p(1-p)n,
\end{align*}
where the first step uses the Cauchy-Schwarz inequality, and the second
step uses the fact that the binomial distribution $B(n,p)$ has variance
$p(1-p)n.$
\end{proof}

\subsection{Approximation by polynomials}

We let $\deg p$ denote the total degree of a multivariate polynomial
$p.$ In this paper, we use the terms ``degree'' and ``total degree''
interchangeably, preferring the former for brevity. Let $\phi\colon X\to\Re$
be given, for a finite subset $X\subset\Re^{n}.$ The \emph{$\epsilon$-approximate
degree} of $\phi,$ denoted $\degeps(\phi),$ is the least degree
of a real polynomial $p$ such that $\|\phi-p\|_{\infty}\leq\epsilon.$
We generalize this definition to partial functions $\phi$ on $X$
by defining $\degeps(\phi)$ as the least degree of a real polynomial
$p$ with 
\begin{align}
\left.\begin{aligned} & |\phi(x)-p(x)|\leq\epsilon, &  & x\in\dom\phi,\\
 & |p(x)|\leq1+\epsilon, &  & x\in X\setminus\dom\phi.
\end{aligned}
\qquad\qquad\qquad\qquad\qquad\qquad\right\} \label{eqn:partial-approx}
\end{align}
For a (possibly partial) real function $\phi$ on a finite subset
$X\subset\Re^{n},$ we define $E(\phi,d)$ to be the least $\epsilon$
such that (\ref{eqn:partial-approx}) holds for some polynomial of
degree at most $d.$ In this notation, $\deg_{\epsilon}(\phi)=\min\{d:E(\phi,d)\leq\epsilon\}.$
 The canonical setting of the error parameter is $\epsilon=1/3,$
which is without loss of generality since the error in a uniform approximation
of a Boolean function can be reduced from any given constant in $(0,1/2)$
to any other constant in $(0,1/2)$ with only a constant-factor increase
in the degree of the approximant. One of the earliest results on the
approximation of Boolean functions by polynomials is the following
seminal theorem due to Nisan and Szegedy~\cite{nisan-szegedy94degree}.
\begin{thm}[Nisan and Szegedy]
\label{thm:nisan-szegedy} 
\begin{align*}
\adeg(\AND_{n})\geq\adeg(\UAND_{n})=\Omega(\sqrt{n}).
\end{align*}
\end{thm}

\subsection{Multiparty communication}

An excellent introduction to communication complexity theory is the
monograph by Kushilevitz and Nisan~\cite{ccbook}. In our overview,
we will limit ourselves to key definitions and notation. This paper
uses the standard model of randomized multiparty communication known
as the \emph{number-on-the-forehead} \emph{model}~\cite{cfl83multiparty}.
Let $F$ be a (possibly partial) Boolean function on $X_{1}\times X_{2}\times\cdots\times X_{k},$
for some finite sets $X_{1},X_{2},\dots,X_{k}.$ The model features
$k$ players. A given input $(x_{1},x_{2},\dots,x_{k})\in X_{1}\times X_{2}\times\cdots\times X_{k}$
is distributed among the players by placing $x_{i}$ on the forehead
of party $i$ (for $i=1,2,\dots,k$). In other words, party $i$ knows
$x_{1},\dots,x_{i-1},x_{i+1},\dots,x_{k}$ but not $x_{i}.$ The players
communicate according to an agreed-upon protocol by writing bits on
a shared blackboard, visible to them all. They additionally have access
to a shared source of random bits which they can use in deciding what
messages to send. Their goal is to accurately compute the value of
$F$ on any given input in the domain of $F.$ An \emph{$\epsilon$-error
communication protocol for $F$} is one which, on every input $(x_{1},x_{2},\dots,x_{k})\in\dom F,$
produces the correct answer $F(x_{1},x_{2},\dots,x_{k})$ with probability
at least $1-\epsilon.$ The \emph{cost} of a communication protocol
is the total number of bits written to the blackboard in the worst
case on any input. The \emph{$\epsilon$-error randomized communication
complexity of $F$}, denoted\emph{ }$R_{\epsilon}(F)$, is the least
cost of an $\epsilon$-error randomized communication protocol for
$F$. As usual, the standard setting of the error parameter is $\epsilon=1/3,$
which is without loss of generality since the error probability in
a communication protocol can be efficiently reduced by running the
protocol several times independently and outputting the majority answer.

The communication problems of interest to us are \emph{generalized
inner product} $\GIP_{n,k}\colon\zoo^{n\times k}\to\zoo$ and \emph{set
disjointness $\DISJ_{n,k}\colon\zoo^{n\times k}\to\zoo$}, given by
\begin{align*}
\GIP_{n,k}(X)= & \bigoplus_{i=1}^{n}\bigwedge_{j=1}^{k}X_{i,j},\\
\DISJ_{n,k}(X)= & \bigwedge_{i=1}^{n}\bigvee_{j=1}^{k}\overline{X}_{i,j}.
\end{align*}
These $k$-party communication problems are both defined on $n\times k$
matrices, where the $i$th party receives as input all but the $i$th
column of the matrix. The disjointness function evaluates to true
if and only if the input matrix does not have an all-ones row, whereas
the generalized inner product function evaluates to true if and only
if the number of all-ones rows is odd. We also consider a partial
Boolean function $\UDISJ_{n,k}$ on $\zoo^{n\times k}$, called \emph{unique
set disjointness} and defined as the restriction of $\DISJ_{n,k}$
to matrices with at most one all-ones row. In other words, $\UDISJ_{n,k}(X)$
is undefined if $X$ has two or more all-ones rows, and is given by
$\UDISJ_{n,k}(X)=\DISJ_{n,k}(X)$ otherwise.

Let $G$ be a (possibly partial) Boolean function on $X_{1}\times X_{2}\times\cdots\times X_{k},$
representing a $k$-party communication problem, and let $f$ be a
(possibly partial) Boolean function on $\zoon.$ We view the composition
$f\circ G$ as a $k$-party communication problem on $X_{1}^{n}\times X_{2}^{n}\times\cdots\times X_{k}^{n}.$
It will be helpful to keep in mind that for all positive integers
$m$ and $n,$ 
\begin{align}
\GIP_{mn,k} & =\XOR_{m}\circ\GIP_{n,k},\label{eq:gip-decompose}\\
\DISJ_{mn,k} & =\AND_{m}\circ\DISJ_{n,k},\label{eq:disj-decompose}\\
\UDISJ_{mn,k} & =\UAND_{m}\circ\UDISJ_{n,k}.\label{eq:udisj-decompose}
\end{align}
Similarly, if $F_{i}$ for $i=1,2,\dots,m$ is a (possibly partial)
$k$-party communication problem on $X_{i,1}\times X_{i,2}\times\cdots\times X_{i,k}$,
we view $\bigoplus_{i=1}^{m}F_{i}$ as a $k$-party communication
problem on $(\prod X_{i,1})\times(\prod X_{i,2})\times\cdots\times(\prod X_{i,k}).$ 

\subsection{Cylinder intersections}

Let $X_{1},X_{2},\dots,X_{k}$ be nonempty finite sets. A \emph{cylinder
intersection} on $X_{1}\times X_{2}\times\cdots\times X_{k}$ is any
function $\chi\colon X_{1}\times X_{2}\times\cdots\times X_{k}\to\zoo$
of the form 
\begin{align}
\chi(x_{1},\dots,x_{k})=\prod_{i=1}^{k}\chi_{i}(x_{1},\dots,x_{i-1},x_{i+1},\dots,x_{k}),\label{eqn:def-cylinder}
\end{align}
where $\chi_{i}\colon X_{1}\times\cdots\times X_{i-1}\times X_{i+1}\times\cdots\times X_{k}\to\zoo.$
In other words, a cylinder intersection is the product of $k$ Boolean
functions, where the $i$th function does not depend on the $i$th
coordinate but may depend arbitrarily on the other $k-1$ coordinates.
 For a given set $S\subseteq\{1,2,\dots,k\},$ we further specialize
this notion to \emph{$S$-cylinder intersections}, defined as functions
of the form
\[
\chi(x_{1},\dots,x_{k})=\prod_{i\in S}\chi_{i}(x_{1},\dots,x_{i-1},x_{i+1},\dots,x_{k})
\]
for some $\chi_{i}\colon X_{1}\times\cdots\times X_{i-1}\times X_{i+1}\times\cdots\times X_{k}\to\zoo$.
Finally, an \emph{$\ell$-cylinder intersection }on $X_{1}\times X_{2}\times\cdots\times X_{k}$
is any $S$-cylinder intersection for a subset $S\subseteq\{1,2,\dots,k\}$
of cardinality at most $\ell.$ Cylinder intersections were introduced
by Babai, Nisan, and Szegedy~\cite{bns92} and play a fundamental
role in the theory due to the following fact.
\begin{fact}
\label{fact:cylinders-breakdown} Let $\Pi\colon X_{1}\times X_{2}\times\cdots\times X_{k}\to\zoo$
be a deterministic $k$-party communication protocol with cost $c.$
Then 
\begin{align*}
\Pi=\sum_{i=1}^{2^{c}}a_{i}\chi_{i}
\end{align*}
for some $\min\{c,k\}$-cylinder intersections $\chi_{1},\dots,\chi_{2^{c}}$
and some $a_{1},\dots,a_{2^{c}}\in\zoo.$
\end{fact}

\noindent We refer the reader to \cite{ccbook} for a simple proof
of Fact~\ref{fact:cylinders-breakdown}. Recall that a randomized
protocol of cost $c$ is a probability distribution on deterministic
protocols of cost $c.$ With this in mind, one easily infers the following
from Fact~\ref{fact:cylinders-breakdown}:
\begin{cor}
\label{cor:cylinders-function} Let $F$ be a $($possibly partial$)$
Boolean function on $X_{1}\times X_{2}\times\cdots\times X_{k}.$
If $R_{\epsilon}(F)=c,$ then 
\begin{align*}
 & |F(x_{1},\dots,x_{k})-\Pi(x_{1},\dots,x_{k})|\leq\epsilon, &  & (x_{1},\dots,x_{k})\in\dom F,\\
 & |\Pi(x_{1},\dots,x_{k})|\leq1, &  & (x_{1},\dots,x_{k})\in X_{1}\times X_{2}\times\cdots\times X_{k},
\end{align*}
where $\Pi=\sum_{\chi}a_{\chi}\chi$ is a linear combination of $\min\{c,k\}$-cylinder
intersections with $\sum_{\chi}|a_{\chi}|\leq2^{c}.$
\end{cor}

\subsection{Discrepancy}

For a (possibly partial) Boolean function $F$ on $X_{1}\times X_{2}\times\cdots\times X_{k}$,
a probability distribution $\mu$ on the domain of $F,$ and a set
$S\subseteq\{1,2,\dots,k\},$ the \emph{$S$-discrepancy} \emph{of
$F$ with respect to $\mu$} is defined as 
\begin{align*}
\disc_{S}(F,\mu) & =\max_{\chi}|\langle(-1)^{F},\mu\cdot\chi\rangle|\\
 & =\max_{\chi}\left|\sum_{x\in\dom F}(-1)^{F(x)}\mu(x)\chi(x)\right|,
\end{align*}
where the maximum is over $S$-cylinder intersections $\chi$. Further
maximizing over $S$ gives the key notions of \emph{$\ell$-discrepancy
}and \emph{discrepancy}, as follows:
\begin{align*}
\disc_{\ell}(F,\mu) & =\max_{\substack{S\subseteq\{1,2,\dots,k\}\\
|S|\leq\ell
}
}\disc_{S}(F,\mu),\\
\disc(F,\mu) & =\max_{S\subseteq\{1,2,\dots,k\}}\disc_{S}(F,\mu).
\end{align*}
By definition, 
\[
\disc_{S}(F,\mu)\leq\disc_{\ell}(F,\mu)\leq\disc(F,\mu)
\]
for every $\ell=0,1,\dots,k$ and every set $S$ of cardinality at
most $\ell.$ 

In light of Corollary~\ref{cor:cylinders-function}, upper bounds
on the discrepancy give lower bounds on the communication complexity.
This fundamental technique is known as the \emph{discrepancy method}~\cite{chor-goldreich88ip,bns92,ccbook}:
\begin{thm}[Discrepancy method]
\label{thm:dm} For every $($possibly partial$)$ Boolean function
$F$ on $X_{1}\times X_{2}\times\cdots\times X_{k}$ and every probability
distribution $\mu$ on the domain of $F,$ 
\begin{align}
 & 2^{R_{\epsilon}(F)}\geq\frac{1-2\epsilon}{\disc(F,\mu)}.\label{eq:dm-general}\\
\intertext{\text{More generally,}} & 2^{R_{\epsilon}(F)}\geq\frac{1-2\epsilon}{\disc_{\min\{R_{\epsilon}(F),k\}}(F,\mu)}.\label{eq:dm-l-party}
\end{align}
\end{thm}

\noindent A proof of (\ref{eq:dm-general}) can be found in \cite{sherstov12mdisj};
that same proof carries over to discrepancy with respect to any given
family $\chi$ of functions and in particular establishes (\ref{eq:dm-l-party})
as well. 

A useful property of discrepancy is its convexity in the second argument,
as formalized by the following proposition. 
\begin{prop}[Convexity of discrepancy]
\label{prop:disc-convex}Let $F$ be a $($possibly partial$)$ Boolean
function on $X_{1}\times X_{2}\times\cdots\times X_{k}$, and let
$\mu$ and $\lambda$ be probability distributions on the domain of
$F.$ Then for every $S\subseteq\{1,2,\dots,k\}$ and $0\leq p\leq1,$
\begin{align*}
\disc_{S}(F,p\mu+(1-p)\lambda) & \leq p\disc_{S}(F,\mu)+(1-p)\disc_{S}(F,\lambda),
\end{align*}
and likewise for $\disc_{\ell}$ and $\disc.$
\end{prop}

\noindent By induction, Proposition~\ref{prop:disc-convex} immediately
generalizes to any finite convex combination of probability distributions.
It is this more general form that we will invoke in our applications. 
\begin{proof}[Proof of Proposition~\emph{\ref{prop:disc-convex}}]
Immediate from the following inequality for any cylinder intersection
$\chi$:
\begin{multline*}
|\langle(-1)^{F},(p\mu+(1-p)\lambda)\cdot\chi\rangle|\\
\leq p\cdot|\langle(-1)^{F},\mu\cdot\chi\rangle|+(1-p)\cdot|\langle(-1)^{F},\lambda\cdot\chi\rangle|,\qquad\tag*{}
\end{multline*}
where for partial functions $F$ the inner products are restricted
to the domain of $F.$
\end{proof}
It is clear that discrepancy is a continuous function of the input
distribution. The following proposition quantifies this continuity.
\begin{prop}[Continuity of discrepancy]
\label{prop:disc-continuous}For any $k$-party communication problem
$F\colon X_{1}\times X_{2}\times\cdots\times X_{k}\to\zoo$ and any
probability distributions $\mu$ and $\tilde{\mu}$ on the domain
of $F,$ 
\[
\disc(F,\mu)\leq\disc(F,\tilde{\mu})+\|\mu-\tilde{\mu}\|_{1}.
\]
More generally, for any $S\subseteq\{1,2,\dots,k\},$ any $($possibly
partial$)$ $k$-party communication problems $F_{1},F_{2},\dots,F_{m},$
and any probability distributions $\mu_{1},\mu_{2},\dots,\mu_{m}$
and $\tilde{\mu}_{1},\tilde{\mu}_{2},\dots,\tilde{\mu}_{m}$ on the
corresponding domains,
\[
\disc_{S}\left(\bigoplus_{i=1}^{m}F_{i},\bigotimes_{i=1}^{m}\mu_{i}\right)\leq\sum_{A\subseteq\{1,2,\dots,m\}}\disc_{S}\left(\bigoplus_{i\in A}F_{i},\bigotimes_{i\in A}\tilde{\mu}_{i}\right)\prod_{i\notin A}\|\mu_{i}-\tilde{\mu}_{i}\|_{1}
\]
and likewise for $\disc_{\ell}$ and $\disc.$
\end{prop}

\begin{proof}
It suffices to prove the claim for $\disc_{S}$. Fix a set $S\subseteq\{1,2,\dots,k\}$
and an $S$-cylinder intersection $\chi$ with
\[
\disc_{S}\left(\bigoplus_{i=1}^{m}F_{i},\bigotimes_{i=1}^{m}\mu_{i}\right)=\left|\left\langle \bigotimes_{i=1}^{m}(-1)^{F_{i}},\chi\cdot\bigotimes_{i=1}^{m}\mu_{i}\right\rangle \right|,
\]
where as usual the inner product on the right-hand side is restricted
to $\prod\dom F_{i}.$ Then
\begin{align}
\disc_{S}\left(\bigoplus_{i=1}^{m}F_{i},\bigotimes_{i=1}^{m}\mu_{i}\right) & =\left|\left\langle \bigotimes_{i=1}^{m}(-1)^{F_{i}},\chi\cdot\bigotimes_{i=1}^{m}(\tilde{\mu}_{i}+(\mu_{i}-\tilde{\mu}_{i}))\right\rangle \right|\nonumber \\
 & =\left|\sum_{A\subseteq\{1,2,\dots,m\}}\left\langle \bigotimes_{i=1}^{m}(-1)^{F_{i}},\chi\cdot\Lambda_{A}\right\rangle \right|,\label{eq:intermediate-disc-cont}
\end{align}
where $\Lambda_{A}$ is given by $\Lambda_{A}(x_{1},x_{2}\ldots,x_{m})=\prod_{i\in A}\tilde{\mu}_{i}(x_{i})\cdot\prod_{i\notin A}(\mu_{i}(x_{i})-\tilde{\mu}_{i}(x_{i})).$
Continuing,
\begin{align}
 & \left|\left\langle \bigotimes_{i=1}^{m}(-1)^{F_{i}},\chi\cdot\Lambda_{A}\right\rangle \right|\nonumber \\
 & \qquad=\left|\sum_{x_{1},\ldots,x_{m}}\chi(x)\prod_{i\in A}(-1)^{F_{i}(x_{i})}\tilde{\mu}_{i}(x_{i})\cdot\prod_{i\notin A}(-1)^{F_{i}(x_{i})}(\mu_{i}(x_{i})-\tilde{\mu}_{i}(x_{i}))\right|\nonumber \\
 & \qquad\leq\sum_{x_{i}:i\notin A}\left|\sum_{x_{i}:i\in A}\chi(x)\prod_{i\in A}(-1)^{F_{i}(x_{i})}\tilde{\mu}_{i}(x_{i})\right|\prod_{i\notin A}|\mu_{i}(x_{i})-\tilde{\mu}_{i}(x_{i})|\nonumber \\
 & \qquad\leq\sum_{x_{i}:i\notin A}\disc_{S}\left(\bigoplus_{i\in A}F_{i},\bigotimes_{i\in A}\tilde{\mu}_{i}\right)\prod_{i\notin A}|\mu_{i}(x_{i})-\tilde{\mu}_{i}(x_{i})|\nonumber \\
 & \qquad=\disc_{S}\left(\bigoplus_{i\in A}F_{i},\bigotimes_{i\in A}\tilde{\mu}_{i}\right)\prod_{i\notin A}\|\mu_{i}-\tilde{\mu}_{i}\|_{1},\label{eq:lambda-A-bound-disc-cont}
\end{align}
where the next to last step is legitimate because for any fixing of
$x_{i}$ for $i\notin A$, the function $\chi$ continues to be an
$S$-cylinder intersection with respect to the remaining coordinates.
In view of~(\ref{eq:intermediate-disc-cont}) and~(\ref{eq:lambda-A-bound-disc-cont}),
the proof is complete.
\end{proof}

\section{\label{sec:gip}Inner product}

In this section, we determine the communication complexity of the
inner product problem $\GIP_{n,k}$ for $k\geq\log n$ players. Our
proofs build on the classic lower and upper bounds for this problem
for $k\leq\log n,$ due to Babai et al.~\cite{bns92} and Grolmusz~\cite{grolmusz94multi-upper},
respectively.

\subsection{\label{sec:gip-cc-lower}Lower bound}

For the lower bound, we use the generalization of the discrepancy
method given by~Theorem~\ref{thm:dm}. We will work with the following
input distribution.
\begin{defn}
Let $\upsilon_{n,k,\ell}$ denote the probability distribution on
$n\times k$ Boolean matrices whereby each row is chosen independently
and uniformly at random from the set $\{u\in\zoo^{k}:|u|\geq k-\ell\}.$ 
\end{defn}

\noindent In particular, $\upsilon_{n,k,k}$ is the uniform probability
distribution on $\zoo^{n\times k}$. In this special case, a strong
upper bound on the discrepancy of generalized inner product was obtained
in the seminal work of Babai, Nisan, and Szegedy~\cite{bns92}.
\begin{thm}[Babai, Nisan, and Szegedy]
\label{thm:BNS}For any positive integers $n$ and $k,$ 
\[
\disc(\GIP_{n,k},\upsilon_{n,k,k})\leq\left(1-\frac{1}{4^{k-1}}\right)^{n}.
\]
\end{thm}

We generalize this discrepancy bound to $\upsilon_{n,k,\ell}$ for
any $\ell.$
\begin{thm}
\label{thm:gip-disc}For any positive integers $n,k,\ell$ with $\ell\leq k,$
\[
\disc_{\ell}(\GIP_{n,k},\upsilon_{n,k,\ell})\leq\left(1-\frac{1}{2^{\ell-1}\binom{k}{\leq\ell}}\right)^{n}.
\]
\end{thm}

\noindent For $\ell=k,$ this discrepancy bound is identical to that
of Theorem~\ref{thm:BNS}. In the setting $\ell\ll k$ of interest
to us, however, the new bound is substantially stronger. 
\begin{proof}[Proof of Theorem~\emph{\ref{thm:gip-disc}}]
Since the communication problem $\GIP_{n,k}$ and the probability
distribution $\upsilon_{n,k,\ell}$ are both symmetric with respect
to the $k$ players, we have
\begin{equation}
\disc_{\ell}(\GIP_{n,k},\upsilon_{n,k,\ell})=\disc_{\{1,2,\dots,\ell\}}(\GIP_{n,k},\upsilon_{n,k,\ell}).\label{eq:gip-unsymmetrize}
\end{equation}
For a given set $S\subseteq\{1,2,\dots,n\}$ and a probability distribution
$\sigma$ on matrices $X\in\zoo^{n\times k},$ let $\sigma|_{S}$
stand for the probability distribution induced by $\sigma$ after
conditioning on the event that $(X_{i,\ell+1},X_{i,\ell+2},\dots,X_{i,k})=(1,1,\dots,1)$
if and only if $i\in S.$ Observe that $\upsilon_{n,k,\ell}|_{S}$
is a probability distribution on matrices $X\in\zoo^{n\times k}$
whereby $X|_{S}$ and $X|_{\overline{S}}$ are distributed independently
such that 
\begin{align}
 & X|_{S,\{1,2,\dots,\ell\}}\sim\upsilon_{|S|,\ell,\ell},\label{eq:X-not-S-active-players-gip}\\
 & X|_{S,\{\ell+1,\ell+2,\dots,k\}}=\begin{bmatrix}1 & 1 & \cdots & 1\\
1 & 1 & \cdots & 1\\
\vdots & \vdots & \ddots & \vdots\\
1 & 1 & \cdots & 1
\end{bmatrix},\label{eq:X-not-S-passive-players-gip}
\end{align}
and $X|_{\overline{S}}$ does not have an all-ones row. In particular,
the rows in $X|_{\overline{S}}$ do not affect the value of the function.
The conditional probability distribution of $X|_{S}$ given any value
of $X|_{\overline{S}}$ is always (\ref{eq:X-not-S-active-players-gip})\textendash (\ref{eq:X-not-S-passive-players-gip}).
Viewing $\upsilon_{n,k,\ell}|_{S}$ as the convex combination of these
conditional probability distributions, corresponding to every possible
value of $X|_{\overline{S}}$, we conclude by Proposition~\ref{prop:disc-convex}
that
\begin{align}
\disc_{\{1,2,\dots,\ell\}}(\GIP_{n,k},\upsilon_{n,k,\ell}|_{S}) & \leq\disc(\GIP_{|S|,\ell},\upsilon_{|S|,\ell,\ell}).\label{eq:gip-condition}
\end{align}
We will now express $\upsilon_{n,k,\ell}$ as a convex combination
of probability distributions $\upsilon_{n,k,\ell}|_{S}$ and use the
convexity of discrepancy to complete the proof. Specifically, let
$p=2^{\ell}/\binom{k}{\leq\ell}.$ Then
\begin{align*}
\disc_{\{1,2,\dots,\ell\}} & (\GIP_{n,k},\upsilon_{n,k,\ell})\\
 & =\disc_{\{1,2,\dots,\ell\}}\left(\GIP_{n,k},\Exp_{s\sim B(n,p)}\Exp_{\substack{S\subseteq\{1,2,\dots,n\}\\
|S|=s
}
}\upsilon_{n,k,\ell}|_{S}\right)\\
 & \tag*{{\text{by definition of \ensuremath{\upsilon_{n,k,\ell}}}}}\rule[-5mm]{0mm}{11mm}\\
 & \leq\Exp_{s\sim B(n,p)}\Exp_{\substack{S\subseteq\{1,2,\dots,n\}\\
|S|=s
}
}\disc_{\{1,2,\dots,\ell\}}(\GIP_{n,k},\upsilon_{n,k,\ell}|_{S})\\
 & \tag*{{\text{by Proposition\,\ref{prop:disc-convex}}}}\rule[-5mm]{0mm}{8mm}\\
 & \leq\Exp_{s\sim B(n,p)}\disc(\GIP_{s,k},\upsilon_{s,\ell,\ell})\tag*{{\text{by \eqref{eq:gip-condition}}}}\\
 & \leq\Exp_{s\sim B(n,p)}\left(1-\frac{1}{4^{\ell-1}}\right)^{s}\tag*{{\text{by Theorem\,\ref{thm:BNS}}}}\\
 & =\sum_{s=0}^{n}\binom{n}{s}\left(1-\frac{1}{4^{\ell-1}}\right)^{s}p^{s}(1-p)^{n-s}\\
 & =\left(1-\frac{p}{4^{\ell-1}}\right)^{n}\\
 & =\left(1-\frac{1}{2^{\ell-1}\binom{k}{\leq\ell}}\right)^{n}.
\end{align*}
In light of (\ref{eq:gip-unsymmetrize}), the proof is complete.
\end{proof}
As a corollary to the new bound on the discrepancy of generalized
inner product, we obtain our claimed communication lower bound for
this function.
\begin{thm}
\label{thm:gip-cc-lower}Abbreviate $R=R_{1/3}(\GIP_{n,k}).$ Then
\begin{equation}
\binom{k}{\mathord\leq R}^{2}R\geq\Omega(n).\label{eq:gip-master-formula}
\end{equation}
In particular,
\begin{align}
R_{1/3}(\GIP_{n,k}) & =\Omega\left(\frac{\log n}{\log\left\lceil 1+\frac{k}{\log n}\right\rceil }+1\right).\label{eq:bns-cc-PS}
\end{align}
\end{thm}

\begin{proof}
We have
\begin{align*}
2^{R} & \geq\frac{1}{3\disc_{\min\{R,k\}}(\GIP_{n,k},\upsilon_{n,k,\min\{R,k\}})} &  & \text{by Theorem\,\ref{thm:dm}}\\
 & \geq\frac{1}{3}\exp\left(\frac{n}{2^{\min\{R,k\}}\binom{k}{\leq\min\{R,k\}}}\right) &  & \text{by Theorem\,\ref{thm:gip-disc}}\\
 & \geq\frac{1}{3}\exp\left(\frac{n}{\binom{k}{\mathord\leq R}^{2}}\right),
\end{align*}
settling (\ref{eq:gip-master-formula}). Now, recall from (\ref{eq:binom-sum-bound})
that 
\[
\binom{k}{\mathord\leq R}\leq\e^{R}\left\lceil \frac{k}{R}\right\rceil ^{R}.
\]
Substituting this estimate in (\ref{eq:gip-master-formula}) gives
$\e^{2R}\lceil k/R\rceil^{2R}R\geq\Omega(n),$ whence (\ref{eq:bns-cc-PS}).
\end{proof}

\subsection{Upper bound}

We now prove a matching upper bound on the communication complexity
of inner product for $k\geq\log n$ players. Our proof is based on
Grolmusz's well-known \emph{deterministic} protocol~\cite{grolmusz94multi-upper}
for this function, which we are able to speed up using shared randomness.
In addition to lower communication cost, the protocol below has the
advantage of being simultaneous, which was not the case in previous
work~\cite{grolmusz94multi-upper,babai-gal-kimmel-lokam95simultaneous-multiparty}.
\begin{thm}
\label{thm:gip-upper-cc}For any $k\geq\log n$ and any constant $\epsilon>0,$
\begin{align*}
R_{\epsilon}(\GIP_{n,k}) & =O\left(\frac{\log n}{\log\left\lceil 1+\frac{k}{\log n}\right\rceil }+1\right).
\end{align*}
Moreover, this upper bound is achieved by a simultaneous protocol.
\end{thm}

\begin{proof}
We first consider the case 
\begin{equation}
k\geq\log3n.\label{eq:k-large}
\end{equation}
Recall that the generalized inner product problem, $\GIP_{n,k},$
is the $k$-party problem of determining whether a given $n\times k$
Boolean matrix $X$ contains an odd number of all-ones rows, where
the $i$th party $(1\leq i\leq k)$ sees all the columns of $X$ except
for the $i$th column. Let $\ell$ denote the smallest natural number,
$0\leq\ell\leq k,$ with the property that
\begin{equation}
\binom{k}{\mathord\leq\ell}\geq3n.\label{eq:ell-defined}
\end{equation}
Such $\ell$ exists by (\ref{eq:k-large}). Moreover, in view of the
lower bound  in~(\ref{eq:binom-sum-bound}), it is straightforward
to verify that
\begin{equation}
\ell=O\left(\frac{\log n}{\log\left\lceil 1+\frac{k}{\log n}\right\rceil }+1\right).\label{eq:ell-not-too-large}
\end{equation}

As the first step of the protocol, the players use their shared randomness
to pick a uniformly random row vector $y\in\zoo^{k}$ with Hamming
weight at least $k-\ell.$ The defining property (\ref{eq:ell-defined})
of $\ell$ ensures that with probability $2/3$ or higher, $y$ is
distinct from every row of the input matrix $X.$ We will prove that
conditioned on this event, the protocol is guaranteed to output the
correct answer. We emphasize that this first step requires no communication.
Indeed, it is our only departure from Grolmusz's protocol~\cite{grolmusz94multi-upper},
where the corresponding vector was computed deterministically and
counted toward the communication cost.

The rest of the analysis is identical to \cite{grolmusz94multi-upper}.
Specifically, by renumbering if necessary the players and the columns
of $X$, we may assume that 
\[
y=(\underbrace{0,0,\dots,0}_{j},1,1,\dots,1)
\]
for some $j\leq\ell.$ Let $n_{i}$ denote the number of rows of $X$
of the form 
\[
(\underbrace{0,0,\dots,0}_{i},1,1,\dots,1).
\]
In this notation, $n_{j}=0$ by the assumption on $y,$ and the objective
of the protocol is to compute the quantity $n_{0}\bmod2.$ For $i=1,2,\dots,j,$
the protocol has the $i$th party broadcast the sum $(n_{i-1}+n_{i})\bmod2,$
which is known to him because he sees all but the $i$th coordinate
of every row of $X$. These $j$ broadcasts are sufficient to compute
the answer since
\begin{align*}
n_{0} & \equiv(n_{0}+n_{1})+(n_{1}+n_{2})+\cdots+(n_{j-1}+n_{j})+n_{j}\pmod2\\
 & \equiv(n_{0}+n_{1})+(n_{1}+n_{2})+\cdots+(n_{j-1}+n_{j})\pmod2.
\end{align*}
Observe that the described protocol is simultaneous, with communication
cost bounded by (\ref{eq:ell-not-too-large}). Its error probability
can be reduced from $1/3$ to any constant $\epsilon>0$ by running
several copies of the protocol in parallel and using the majority
answer. 

We handle the case $k\in[\log n,\log3n)$ in a manner analogous to
\cite{grolmusz94multi-upper}, using the composed structure (\ref{eq:gip-decompose})
of generalized inner product. Specifically, the players partition
the input matrix horizontally into submatrices with at most $n/3$
rows each and run the above protocol on the resulting submatrices
with a small constant error parameter, simultaneously and in parallel.
The protocol output is the XOR of these answers.
\end{proof}

\section{\label{sec:disj}Set disjointness}

We now turn to the set disjointness problem $\DISJ_{n,k}$, proving
matching lower and upper bounds on its communication complexity for
$k\geq\log n$ players. Analogous to the previous section, our lower
bound is a reduction to the case $k\leq\log n$ followed by an appeal
to a known lower bound for that setting~\cite{sherstov12mdisj}.
The treatment here is more technical than for inner product.

\subsection{$\ell$-discrepancy\label{sec:discrepancy-of-DISJ}}

For positive integers $n$ and $k$, let $\mu_{n,k}$ denote the uniform
probability distribution on matrices $X\in\zoo^{n\times k}$ such
that $X_{i,1}=X_{i,2}=\cdots=X_{i,k-1}=1$ for precisely one row $i.$
The following result~\cite[Theorem~4.2]{sherstov12mdisj} bounds
the multiparty discrepancy of the XOR of $m$ independent instances
of the set disjointness problem, each distributed according to the
probability distribution just defined.
\begin{thm}[Sherstov]
\label{thm:asymmetric-disj}For any integers $n_{1},n_{2},\dots,n_{m},$
\[
\disc\left(\bigoplus_{i=1}^{m}\DISJ_{n_{i},k},\bigotimes_{i=1}^{m}\mu_{n_{i},k}\right)\leq\frac{(2^{k-1}-1)^{m}}{\sqrt{n_{1}n_{2}\cdots n_{m}}}.
\]
\end{thm}

\noindent For the purposes of this paper, we will slightly adapt Theorem~\ref{thm:asymmetric-disj}
to obtain a discrepancy bound under a more symmetric distribution.
Specifically, let 
\begin{equation}
\sigma_{n,k}=\frac{1}{2}\sigma_{n,k}^{0}+\frac{1}{2}\sigma_{n,k}^{1},\label{eq:sigma-n-k-definition}
\end{equation}
where $\sigma_{n,k}^{0}$ is the uniform probability distribution
on $n\times k$ Boolean matrices without an all-ones row, and $\sigma_{n,k}^{1}$
is the uniform probability distribution on $n\times k$ Boolean matrices
with precisely one all-ones row.
\begin{thm}
\label{thm:symmetric-disj}For any integers $n_{1},n_{2},\dots,n_{m},$
\begin{multline*}
\disc\left(\bigoplus_{i=1}^{m}\DISJ_{n_{i},k},\bigotimes_{i=1}^{m}\sigma_{n_{i},k}\right)\\
\leq\left(\frac{(\sqrt{2^{k}-1}+1)\sqrt{2^{k}-2}}{2}\right)^{m}\frac{1}{\sqrt{n_{1}n_{2}\cdots n_{m}}}.\qquad
\end{multline*}
\end{thm}

\begin{proof}
Instead of analyzing the discrepancy with respect to the tensor product
$\bigotimes\sigma_{n_{i},k},$ we will define a different family of
probability distributions $\tilde{\sigma}_{n_{i},k}$ and bound the
discrepancy with respect to $\bigotimes\tilde{\sigma}_{n_{i},k}$
using Theorem~\ref{thm:asymmetric-disj}. We will then prove that
$\sigma_{n_{i},k}$ and $\tilde{\sigma}_{n_{i},k}$ are close in statistical
distance for each $i$, and appeal to Proposition~\ref{prop:disc-continuous}
to complete the proof.

Abbreviate $p=1/(2^{k}-1).$ For a given set $S\subseteq\{1,2,\dots,n\}$
and a probability distribution $\sigma$ on matrices $X\in\zoo^{n\times k},$
let $\sigma|_{S}$ stand for the probability distribution induced
by $\sigma$ after conditioning on the event that $(X_{i,1},X_{i,2},\dots,X_{i,k})=(1,1,\dots,1,0)$
if and only if $i\in S.$ Observe that for a fixed set $S\subsetneq\{1,2,\dots,n\},$
the convex combination
\begin{equation}
\frac{\sigma_{n,k}^{1}|_{S}}{2}+\Exp_{\substack{S'\supset S\\
|S'|=|S|+1
}
}\frac{\sigma_{n,k}^{0}|_{S'}}{2}\label{eq:sigmas-vs-mu}
\end{equation}
is a probability distribution on matrices $X\in\zoo^{n\times k}$
whereby 
\begin{align*}
X|_{S} & =\begin{bmatrix}1 & 1 & \cdots & 1 & 0\\
1 & 1 & \cdots & 1 & 0\\
\vdots & \vdots & \ddots & \vdots & \vdots\\
1 & 1 & \cdots & 1 & 0
\end{bmatrix}\intertext{and}\\
X|_{\overline{S}} & \sim\mu_{n-|S|,k}.
\end{align*}
In other words, the rows with indices in $S$ are fixed to non-$1^{k}$
values and can therefore be disregarded from the point of view of
set disjointness, whereas the remaining rows have joint probability
distribution $\mu_{n-|S|,k}$. It follows that for any sets $S_{1},S_{2},\dots,S_{m}$
with $S_{i}\subsetneq\{1,2,\dots,n_{i}\},$
\begin{align*}
\disc & \left(\bigoplus_{i=1}^{m}\DISJ_{n_{i},k},\bigotimes_{i=1}^{m}\left(\frac{\sigma_{n_{i},k}^{1}|_{S_{i}}}{2}+\Exp_{\substack{S_{i}'\supset S_{i}\\
|S_{i}'|=|S_{i}|+1
}
}\frac{\sigma_{n_{i},k}^{0}|_{S_{i}'}}{2}\right)\right)\\
 & \qquad\qquad\qquad\qquad\qquad\qquad=\disc\left(\bigoplus_{i=1}^{m}\DISJ_{n_{i}-|S_{i}|,k},\bigotimes_{i=1}^{m}\mu_{n_{i}-|S_{i}|,k}\right)\\
 & \qquad\qquad\qquad\qquad\qquad\qquad\leq\prod_{i=1}^{m}\frac{2^{k-1}-1}{\sqrt{n_{i}-|S_{i}|}}\\
 & \qquad\qquad\qquad\qquad\qquad\qquad=\prod_{i=1}^{m}\frac{1-p}{2p\sqrt{n_{i}-|S_{i}|}},
\end{align*}
where the second step uses Theorem~\ref{thm:asymmetric-disj}. Proposition~\ref{prop:disc-convex}
now implies that for any integers $s_{1},s_{2},\dots,s_{m}$ with
$0\leq s_{i}<n_{i},$ 
\begin{align}
\disc & \left(\bigoplus_{i=1}^{m}\DISJ_{n_{i},k},\bigotimes_{i=1}^{m}\left(\Exp_{|S_{i}|=s_{i}}\frac{\sigma_{n_{i},k}^{1}|_{S_{i}}}{2}+\Exp_{|S_{i}|=s_{i}+1}\frac{\sigma_{n_{i},k}^{0}|_{S_{i}}}{2}\right)\right)\nonumber \\
 & \qquad\qquad\qquad\qquad\qquad\qquad\leq\prod_{i=1}^{m}\frac{1-p}{2p\sqrt{n_{i}-s_{i}}}.\label{eq:disc-mu-tilde-elementary}
\end{align}

We now define an approximation to the probability distribution $\sigma_{n,k}$,
namely, 
\begin{equation}
\tilde{\sigma}_{n,k}=\Exp_{s\sim B(n-1,p)}\left[\Exp_{|S|=s}\frac{\sigma_{n,k}^{1}|_{S}}{2}+\Exp_{|S|=s+1}\frac{\sigma_{n,k}^{0}|_{S}}{2}\right].\label{eq:sigma-tilde-n-k-definition}
\end{equation}
For random integers $s_{1},s_{2},\dots,s_{m}$ distributed independently
according to $s_{i}\sim B(n_{i}-1,p),$ 
\begin{align*}
\!\!\!\!\!\!\disc & \left(\bigoplus_{i=1}^{m}\DISJ_{n_{i},k},\bigotimes_{i=1}^{m}\tilde{\sigma}_{n_{i},k}\right)\\
 & \leq\Exp_{s_{1},\dots,s_{m}}\;\disc\left(\bigoplus_{i=1}^{m}\DISJ_{n_{i},k},\bigotimes_{i=1}^{m}\left(\Exp_{|S_{i}|=s_{i}}\frac{\sigma_{n_{i},k}^{1}|_{S_{i}}}{2}+\Exp_{|S_{i}|=s_{i}+1}\frac{\sigma_{n_{i},k}^{0}|_{S_{i}}}{2}\right)\right)\\
 & \tag*{{\text{by Proposition\,\ref{prop:disc-convex}}}}\\
 & \leq\Exp_{s_{1},\dots,s_{m}}\;\prod_{i=1}^{m}\frac{1-p}{2p\sqrt{n_{i}-s_{i}}}\tag*{{\text{by \eqref{eq:disc-mu-tilde-elementary}}}}\\
 & =\prod_{i=1}^{m}\Exp_{s_{i}}\;\frac{1-p}{2p\sqrt{n_{i}-s_{i}}}\tag*{{\text{by independence}}}\\
 & \leq\prod_{i=1}^{m}\frac{\sqrt{1-p}}{2p\sqrt{n_{i}}}\tag*{{\text{by Fact\,\ref{fact:binomial-expectations}.}}}
\end{align*}
Of course, this calculation shows more generally that
\begin{equation}
\disc\left(\bigoplus_{i\in A}\DISJ_{n_{i},k},\bigotimes_{i\in A}\tilde{\sigma}_{n_{i},k}\right)\leq\prod_{i\in A}\frac{\sqrt{1-p}}{2p\sqrt{n_{i}}}\label{eq:disc-under-tilde-sigma}
\end{equation}
for any set $A\subseteq\{1,2,\dots,m\}.$ This completes the first
part of the program.

We proceed to bound the statistical distance between $\sigma_{n,k}$
and $\tilde{\sigma}_{n,k}.$ To start with,
\begin{align*}
\sigma_{n,k}^{0} & =\sum_{s=0}^{n}\binom{n}{s}p^{s}(1-p)^{n-s}\Exp_{|S|=s}\sigma_{n,k}^{0}|_{S},\\
\sigma_{n,k}^{1} & =\sum_{s=0}^{n-1}\binom{n-1}{s}p^{s}(1-p)^{n-1-s}\Exp_{|S|=s}\sigma_{n,k}^{1}|_{S},\\
\tilde{\sigma}_{n,k} & =\frac{1}{2}\sum_{s=1}^{n}\binom{n-1}{s-1}p^{s-1}(1-p)^{n-s}\Exp_{|S|=s}\sigma_{n,k}^{0}|_{S}\\
 & \qquad\qquad+\frac{1}{2}\sum_{s=0}^{n-1}\binom{n-1}{s}p^{s}(1-p)^{n-1-s}\Exp_{|S|=s}\sigma_{n,k}^{1}|_{S},
\end{align*}
where the first two equations hold by definition, and the third is
a restatement of (\ref{eq:sigma-tilde-n-k-definition}). Then
\begin{align}
\|\sigma_{n,k}-\tilde{\sigma}_{n,k}\|_{1} & =\left\Vert \frac{\sigma_{n,k}^{0}+\sigma_{n,k}^{1}}{2}-\tilde{\sigma}_{n,k}\right\Vert _{1}\nonumber \\
 & =\frac{1}{2}\left\Vert \sum_{s=0}^{n}\binom{n}{s}p^{s}(1-p)^{n-s}\Exp_{|S|=s}\sigma_{n,k}^{0}|_{S}\right.\nonumber \\
 & \qquad\qquad\left.-\sum_{s=1}^{n}\binom{n-1}{s-1}p^{s-1}(1-p)^{n-s}\Exp_{|S|=s}\sigma_{n,k}^{0}|_{S}\right\Vert _{1}\nonumber \\
 & =\frac{1}{2}\sum_{s=1}^{n}\left|\binom{n}{s}p^{s}(1-p)^{n-s}-\binom{n-1}{s-1}p^{s-1}(1-p)^{n-s}\right|\nonumber \\
 & \qquad\qquad+\frac{(1-p)^{n}}{2}\nonumber \\
 & =\frac{1}{2pn}\sum_{s=0}^{n}\binom{n}{s}p^{s}(1-p)^{n-s}|s-pn|\nonumber \\
 & =\frac{1}{2pn}\Exp_{s\sim B(n,p)}\;|s-pn|\nonumber \\
 & \leq\frac{1}{2}\sqrt{\frac{1-p}{pn}},\label{eq:sigma-sigma-tilde-are-close}
\end{align}
where the first and last steps use (\ref{eq:sigma-n-k-definition})
and Fact~\ref{fact:binomial-expectations}, respectively. This completes
the second part of the proof program.

It remains to put together the above ingredients and appeal to the
continuity of discrepancy: 
\begin{align*}
\disc & \left(\bigoplus_{i=1}^{m}\DISJ_{n_{i},k},\bigotimes_{i=1}^{m}\sigma_{n_{i},k}\right)\\
 & \qquad\;\leq\sum_{A\subseteq\{1,2,\dots,m\}}\disc\left(\bigoplus_{i\in A}\DISJ_{n_{i},k},\bigotimes_{i\in A}\tilde{\sigma}_{n_{i},k}\right)\prod_{i\notin A}\|\sigma_{n_{i},k}-\tilde{\sigma}_{n_{i},k}\|_{1}\\
 & \qquad\;\leq\sum_{A\subseteq\{1,2,\dots,m\}}\;\prod_{i\in A}\frac{\sqrt{1-p}}{2p\sqrt{n_{i}}}\cdot\prod_{i\notin A}\frac{\sqrt{1-p}}{2\sqrt{pn_{i}}}\\
 & \qquad\;=\prod_{i=1}^{m}\left(\frac{\sqrt{1-p}}{2p\sqrt{n_{i}}}+\frac{\sqrt{1-p}}{2\sqrt{pn_{i}}}\right)\\
 & \qquad\;=\prod_{i=1}^{m}\frac{(\sqrt{2^{k}-1}+1)\sqrt{2^{k}-2}}{2\sqrt{n_{i}}},
\end{align*}
where the first inequality uses Proposition~\ref{prop:disc-continuous},
and the second inequality follows from the upper bounds in (\ref{eq:disc-under-tilde-sigma})
and (\ref{eq:sigma-sigma-tilde-are-close}).
\end{proof}
We now refine the previous theorem for cylinder intersections that
contain significantly fewer cylinders than there are players. Formally,
define $\sigma_{n,k,\ell}^{0}$ to be the probability distribution
on $n\times k$ Boolean matrices whereby the rows are chosen independently
and uniformly at random from the set $\{u\in\zoo^{k}:k-\ell\leq|u|\leq k-1\}.$
Define $\sigma_{n,k,\ell}^{1}$ to be the probability distribution
on $n\times k$ Boolean matrices whereby a randomly chosen row is
set to $1^{k}$ and the remaining rows are chosen independently and
uniformly at random from the set $\{u\in\zoo^{k}:k-\ell\leq|u|\leq k-1\}.$
We will analyze the $\ell$-discrepancy of set disjointness with respect
to the probability distribution 
\begin{equation}
\sigma_{n,k,\ell}=\frac{1}{2}\sigma_{n,k,\ell}^{0}+\frac{1}{2}\sigma_{n,k,\ell}^{1}.\label{eq:sigma-n-k-l-definition}
\end{equation}
This setup is indeed a generalization of the case $\ell=k$ dealt
with above. Specifically,  
\begin{align*}
\sigma_{n,k,k}^{0} & =\sigma_{n,k}^{0},\\
\sigma_{n,k,k}^{1} & =\sigma_{n,k}^{1},\\
\sigma_{n,k,k} & =\sigma_{n,k}.
\end{align*}
\begin{thm}
\label{thm:symmetric-disj-few-parties}For any integers $n_{1},n_{2},\dots,n_{m}\geq1$
and $k\geq\ell\geq1,$ 
\begin{multline*}
\disc_{\ell}\left(\bigoplus_{i=1}^{m}\DISJ_{n_{i},k},\bigotimes_{i=1}^{m}\sigma_{n_{i},k,\ell}\right)\\
\leq(2^{\ell}-1)^{m/2}\left(\binom{k}{1}+\binom{k}{2}+\cdots+\binom{k}{\ell}\right)^{m/2}\frac{1}{\sqrt{n_{1}n_{2}\cdots n_{m}}}.
\end{multline*}
\end{thm}

\noindent In the setting $\ell\ll k$ of interest to us, the new bound
is considerably stronger than the bound of Theorem~\ref{thm:symmetric-disj}.
The proof is essentially a reprise of Theorem~\ref{thm:symmetric-disj}.
Indeed, we could have combined the two theorems for a more economical
presentation. Treating them separately, as we do in this paper, has
the advantage of simplifying the notation and illustrating the proof
idea in a simpler setting first. 
\begin{proof}[Proof of Theorem~\emph{\ref{thm:symmetric-disj-few-parties}}]
We will closely follow the proof of the previous theorem. Specifically,
instead of analyzing the discrepancy with respect to the tensor product
$\bigotimes\sigma_{n_{i},k,\ell},$ we will define a different family
of probability distributions $\tilde{\sigma}_{n_{i},k,\ell}$ and
bound the discrepancy with respect to $\bigotimes\tilde{\sigma}_{n_{i},k,\ell}$.
We will then prove that $\sigma_{n_{i},k,\ell}$ and $\tilde{\sigma}_{n_{i},k,\ell}$
are close in statistical distance for each $i$, and appeal to Proposition~\ref{prop:disc-continuous}
to complete the proof.

Since the communication problem $\DISJ_{n,k}$ and the probability
distribution $\sigma_{n,k,\ell}$ are both symmetric with respect
to the $k$ players, we have
\begin{multline}
\disc_{\ell}\left(\bigoplus_{i=1}^{m}\DISJ_{n_{i},k},\bigotimes_{i=1}^{m}\sigma_{n_{i},k,\ell}\right)\\
=\disc_{\{1,2,\dots,\ell\}}\left(\bigoplus_{i=1}^{m}\DISJ_{n_{i},k},\bigotimes_{i=1}^{m}\sigma_{n_{i},k,\ell}\right).\label{eq:maximizing-cylinder-intersection}
\end{multline}
\begin{comment}
In other words, the $\ell$-party discrepancy is achieved for a cylinder
intersection $\chi$ that corresponds to the first $\ell$ players,
with
\[
\chi=\prod_{i=1}^{\ell}\chi_{i}
\]
for some functions $\chi_{i}\colon(\zoo^{n_{1}+n_{2}+\cdots+n_{m}})^{k}\to\zoo$
of $m$ Boolean matrices of orders $n_{1}\times k,\dots,n_{m}\times k$
that do not depend on the $i$th row of the matrices. Clearly, $\chi$
is a $\ell$-dimensional cylinder intersection $(\zoo^{n_{1}+n_{2}+\cdots+n_{m}})^{\ell}\to\zoo$
with respect to the first $\ell$ rows of the input matrices for any
fixed assignment to the remaining rows. 
\end{comment}
For a given set $S\subseteq\{1,2,\dots,n\}$ and a probability distribution
$\sigma$ on matrices $X\in\zoo^{n\times k},$ let $\sigma|_{S}$
stand for the probability distribution induced by $\sigma$ after
conditioning on the event that $(X_{i,\ell+1},X_{i,\ell+2},\dots,X_{i,k})=(1,1,\dots,1)$
if and only if $i\in S.$ Observe that for a fixed nonempty set $S\subseteq\{1,2,\dots,n\},$
the convex combination
\begin{equation}
\frac{\sigma_{n,k,\ell}^{1}|_{S}}{2}+\frac{\sigma_{n,k,\ell}^{0}|_{S}}{2}\label{eq:sigmas-vs-mu-1}
\end{equation}
is the probability distribution on matrices $X\in\zoo^{n\times k}$
whereby $X|_{S}$ and $X|_{\overline{S}}$ are distributed independently
such that 
\begin{align}
 & X|_{S,\{1,2,\dots,\ell\}}\sim\sigma_{|S|,\ell},\label{eq:X-not-S-active-players}\\
 & X|_{S,\{\ell+1,\ell+2,\dots,k\}}=\begin{bmatrix}1 & 1 & \cdots & 1\\
1 & 1 & \cdots & 1\\
\vdots & \vdots & \ddots & \vdots\\
1 & 1 & \cdots & 1
\end{bmatrix},\label{eq:X-not-S-passive-players}
\end{align}
and $X|_{\overline{S}}$ does not have an all-ones row. In particular,
the rows in $X|_{\overline{S}}$ do not affect the value of the function.
Since the conditional probability distribution of $X|_{S}$ given
$X|_{\overline{S}}$ is always (\ref{eq:X-not-S-active-players})\textendash (\ref{eq:X-not-S-passive-players}),
we arrive at the following conclusion:  for any nonempty sets $S_{1},S_{2},\dots,S_{m}$
with $S_{i}\subseteq\{1,2,\dots,n_{i}\},$
\begin{align*}
\disc_{\{1,2,\dots,\ell\}} & \left(\bigoplus_{i=1}^{m}\DISJ_{n_{i},k},\bigotimes_{i=1}^{m}\left(\frac{\sigma_{n_{i},k,\ell}^{1}|_{S_{i}}}{2}+\frac{\sigma_{n_{i},k,\ell}^{0}|_{S_{i}}}{2}\right)\right)\\
 & \qquad\qquad\qquad\qquad\qquad\leq\disc\left(\bigoplus_{i=1}^{m}\DISJ_{|S_{i}|,\ell},\bigotimes_{i=1}^{m}\sigma_{|S_{i}|,\ell}\right)\\
 & \qquad\qquad\qquad\qquad\qquad\leq\prod_{i=1}^{m}\frac{(\sqrt{2^{\ell}-1}+1)\sqrt{2^{\ell}-2}}{2\sqrt{|S_{i}|}},
\end{align*}
where the second step holds by Theorem~\ref{thm:symmetric-disj}.
Proposition~\ref{prop:disc-convex} now implies that for any integers
$s_{1},s_{2},\dots,s_{m}$ with $1\leq s_{i}\leq n_{i},$ 
\begin{align}
\disc & _{\{1,2,\dots,\ell\}}\left(\bigoplus_{i=1}^{m}\DISJ_{n_{i},k},\bigotimes_{i=1}^{m}\left(\Exp_{|S_{i}|=s_{i}}\frac{\sigma_{n_{i},k,\ell}^{1}|_{S_{i}}}{2}+\Exp_{|S_{i}|=s_{i}}\frac{\sigma_{n_{i},k,\ell}^{0}|_{S_{i}}}{2}\right)\right)\nonumber \\
 & \qquad\qquad\qquad\qquad\qquad\qquad\leq\prod_{i=1}^{m}\frac{(\sqrt{2^{\ell}-1}+1)\sqrt{2^{\ell}-2}}{2\sqrt{s_{i}}}.\label{eq:disc-mu-tilde-elementary-1}
\end{align}

We now define the promised approximation $\tilde{\sigma}_{n,k,\ell}$
to the probability distribution $\sigma_{n,k,\ell}$, namely, 
\begin{equation}
\tilde{\sigma}_{n,k,\ell}=\Exp_{s\sim B(n-1,q)}\;\Exp_{|S|=s+1}\frac{\sigma_{n,k,\ell}^{0}|_{S}+\sigma_{n,k,\ell}^{1}|_{S}}{2}\label{eq:sigma-tilde-n-k-definition-1}
\end{equation}
where 
\[
q=\frac{2^{\ell}-1}{\binom{k}{1}+\binom{k}{2}+\cdots+\binom{k}{\ell}}.
\]
Then for random integers $s_{1},s_{2},\dots,s_{m}$ distributed independently
according to $s_{i}\sim B(n_{i}-1,q),$ we have
\begin{align*}
\!\!\!\!\!\!\disc & _{\{1,2,\dots,\ell\}}\left(\bigoplus_{i=1}^{m}\DISJ_{n_{i},k},\bigotimes_{i=1}^{m}\tilde{\sigma}_{n_{i},k,\ell}\right)\\
 & \leq\Exp_{s_{1},\dots,s_{m}}\;\disc_{\{1,2,\dots,\ell\}}\left(\bigoplus_{i=1}^{m}\DISJ_{n_{i},k},\bigotimes_{i=1}^{m}\left(\Exp_{|S_{i}|=s_{i}+1}\frac{\sigma_{n_{i},k,\ell}^{0}|_{S_{i}}+\sigma_{n_{i},k,\ell}^{1}|_{S_{i}}}{2}\right)\right)\\
 & \tag*{{\text{by Proposition\,\ref{prop:disc-convex}}}}\\
 & \leq\Exp_{s_{1},\dots,s_{m}}\;\prod_{i=1}^{m}\frac{(\sqrt{2^{\ell}-1}+1)\sqrt{2^{\ell}-2}}{2\sqrt{s_{i}+1}}\tag*{{\text{by \eqref{eq:disc-mu-tilde-elementary-1}}}}\\
 & =\prod_{i=1}^{m}\Exp_{s_{i}}\;\frac{(\sqrt{2^{\ell}-1}+1)\sqrt{2^{\ell}-2}}{2\sqrt{s_{i}+1}}\tag*{{\text{by independence}}}\\
 & \leq\prod_{i=1}^{m}\frac{(\sqrt{2^{\ell}-1}+1)\sqrt{2^{\ell}-2}}{2\sqrt{qn_{i}}}\tag*{{\text{by Fact\,\ref{fact:binomial-expectations}.}}}
\end{align*}
Of course, this calculation shows more generally that
\begin{multline}
\disc_{\{1,2,\dots,\ell\}}\left(\bigoplus_{i\in A}\DISJ_{n_{i},k},\bigotimes_{i\in A}\tilde{\sigma}_{n_{i},k,\ell}\right)\\
\leq\prod_{i\in A}\frac{(\sqrt{2^{\ell}-1}+1)\sqrt{2^{\ell}-2}}{2\sqrt{qn_{i}}}\qquad\label{eq:disc-under-tilde-sigma-few-parties}
\end{multline}
for any set $A\subseteq\{1,2,\dots,m\}.$ This completes the first
part of the program.

In this second part of the proof, we bound the statistical distance
between $\sigma_{n,k,\ell}$ and $\tilde{\sigma}_{n,k,\ell}.$ We
have
\begin{align*}
\sigma_{n,k,\ell}^{0} & =\sum_{s=0}^{n}\binom{n}{s}q^{s}(1-q)^{n-s}\Exp_{|S|=s}\sigma_{n,k,\ell}^{0}|_{S},\\
\sigma_{n,k,\ell}^{1} & =\sum_{s=1}^{n}\binom{n-1}{s-1}q^{s-1}(1-q)^{n-s}\Exp_{|S|=s}\sigma_{n,k,\ell}^{1}|_{S},\\
\tilde{\sigma}_{n,k,\ell} & =\frac{1}{2}\sum_{s=1}^{n}\binom{n-1}{s-1}q^{s-1}(1-q)^{n-s}\Exp_{|S|=s}\sigma_{n,k,\ell}^{0}|_{S}\\
 & \qquad\qquad+\frac{1}{2}\sum_{s=1}^{n}\binom{n-1}{s-1}q^{s-1}(1-q)^{n-s}\Exp_{|S|=s}\sigma_{n,k,\ell}^{1}|_{S},
\end{align*}
where the first two equations hold by definition, and the third is
a restatement of (\ref{eq:sigma-tilde-n-k-definition-1}). Then
\begin{align}
\|\sigma_{n,k,\ell}-\tilde{\sigma}_{n,k,\ell}\|_{1} & =\left\Vert \frac{\sigma_{n,k,\ell}^{1}+\sigma_{n,k,\ell}^{0}}{2}-\tilde{\sigma}_{n,k,\ell}\right\Vert _{1}\nonumber \\
 & =\frac{1}{2}\left\Vert \sum_{s=0}^{n}\binom{n}{s}q^{s}(1-q)^{n-s}\Exp_{|S|=s}\sigma_{n,k,\ell}^{0}|_{S}\right.\nonumber \\
 & \qquad\qquad\left.-\sum_{s=1}^{n}\binom{n-1}{s-1}q^{s-1}(1-q)^{n-s}\Exp_{|S|=s}\sigma_{n,k,\ell}^{0}|_{S}\right\Vert _{1}\nonumber \\
 & =\frac{1}{2}\sum_{s=1}^{n}\left|\binom{n}{s}q^{s}(1-q)^{n-s}-\binom{n-1}{s-1}q^{s-1}(1-q)^{n-s}\right|\nonumber \\
 & \qquad\qquad+\frac{(1-q)^{n}}{2}\nonumber \\
 & =\frac{1}{2qn}\sum_{s=0}^{n}\binom{n}{s}q^{s}(1-q)^{n-s}|s-qn|\nonumber \\
 & =\frac{1}{2qn}\Exp_{s\sim B(n,q)}\;|s-qn|\nonumber \\
 & \leq\frac{1}{2}\sqrt{\frac{1-q}{qn}},\label{eq:sigma-sigma-tilde-are-close-few-parties}
\end{align}
where the first and last steps use (\ref{eq:sigma-n-k-definition})
and Fact~\ref{fact:binomial-expectations}, respectively. This completes
the second part of the proof program.

It remains to put together the above ingredients and appeal to the
continuity of discrepancy: 
\begin{align*}
\disc_{\ell} & \left(\bigoplus_{i=1}^{m}\DISJ_{n_{i},k},\bigotimes_{i=1}^{m}\sigma_{n_{i},k,\ell}\right)\\
 & \;\;\;\;=\disc{}_{\{1,2,\dots,\ell\}}\left(\bigoplus_{i=1}^{m}\DISJ_{n_{i},k},\bigotimes_{i=1}^{m}\sigma_{n_{i},k,\ell}\right)\\
 & \;\;\;\;\leq\sum_{A\subseteq\{1,2,\dots,m\}}\disc_{\{1,2,\dots,\ell\}}\left(\bigoplus_{i\in A}\DISJ_{n_{i},k},\bigotimes_{i\in A}\tilde{\sigma}_{n_{i},k,\ell}\right)\\
 & \qquad\qquad\qquad\qquad\qquad\qquad\qquad\qquad\qquad\quad\times\prod_{i\notin A}\|\sigma_{n_{i},k,\ell}-\tilde{\sigma}_{n_{i},k,\ell}\|_{1}\\
 & \;\;\;\;\leq\sum_{A\subseteq\{1,2,\dots,m\}}\;\prod_{i\in A}\frac{(\sqrt{2^{\ell}-1}+1)\sqrt{2^{\ell}-2}}{2\sqrt{qn_{i}}}\cdot\prod_{i\notin A}\frac{\sqrt{1-q}}{2\sqrt{qn_{i}}}\\
 & \;\;\;\;=\prod_{i=1}^{m}\frac{(\sqrt{2^{\ell}-1}+1)\sqrt{2^{\ell}-2}+\sqrt{1-q}}{2\sqrt{qn_{i}}}\\
 & \;\;\;\;\leq\prod_{i=1}^{m}\frac{(2^{\ell}-1)}{\sqrt{qn_{i}}},
\end{align*}
where the first step is valid by (\ref{eq:maximizing-cylinder-intersection}),
the next step uses Proposition~\ref{prop:disc-continuous}, and the
third step follows from the upper bounds in (\ref{eq:disc-under-tilde-sigma-few-parties})
and (\ref{eq:sigma-sigma-tilde-are-close-few-parties}).
\end{proof}

\subsection{Lower bound}

We are now in a position to prove the claimed lower bound on the communication
complexity of set disjointness. Following~\cite{sherstov12mdisj},
we find it helpful to work in the more general setting of composed
communication problems $f\circ G,$ where $f$ is any Boolean function
with high approximate degree and $G$ is an instance of set disjointness
on a small number of variables. This approach is motivated by the
composed structure (\ref{eq:disj-decompose})\textendash (\ref{eq:udisj-decompose})
of the set disjointness problem. 

The following communication lower bound was obtained in~\cite[Theorem~5.1]{sherstov12mdisj}.
\begin{thm}[Sherstov]
\label{thm:randomized-S12} Let $f$ be a $($possibly partial$)$
Boolean function on $\zoon.$ Consider the $k$-party communication
problem $F=f\circ\UDISJ_{r,k}.$ Then for $\epsilon,\delta\geq0,$
\begin{align*}
2^{R_{\epsilon}(F)}\geq(\delta-\epsilon)\left(\frac{\deg_{\delta}(f)\sqrt{r}}{2^{k}\e n}\right)^{\deg_{\delta}(f)}.
\end{align*}
\end{thm}

\noindent Using the new discrepancy upper bound in this paper, we
are able to obtain the following improvement: 
\begin{thm}
\label{thm:randomized-VP} Let $f$ be a $($possibly partial$)$
Boolean function on $\zoon.$ Consider the $k$-party communication
problem $F=f\circ\UDISJ_{r,k}.$ Then for $\epsilon,\delta\geq0,$
\begin{align*}
2^{R_{\epsilon}(F)}\geq(\delta-\epsilon)\left(\frac{\deg_{\delta}(f)\sqrt{r}}{2\binom{k}{\leq R_{\epsilon}(F)}\e n}\right)^{\deg_{\delta}(f)}.
\end{align*}
\end{thm}

\noindent The lower bound of Theorem~\ref{thm:randomized-VP} is
always at least as strong as that of Theorem~\ref{thm:randomized-S12},
up to a small constant factor in the denominator. The improvement
becomes significant in the setting $k\gg R_{\epsilon}(F)$ of interest
to us, where the number of players far exceeds the communication requirements
of the problem. 
\begin{proof}[Proof of Theorem~\emph{\ref{thm:randomized-VP}}]
 The proof is virtually identical to that in \cite{sherstov12mdisj},
the only difference being the use of the improved discrepancy upper
bound in this paper (Theorem~\ref{thm:symmetric-disj-few-parties}).
The idea behind the proof is to show that any low-cost communication
protocol for $F$ can be converted into a low-degree polynomial approximating
$f$ in the infinity norm. Details follow.

Put $\ell=\min\{R_{\epsilon}(F),k\}$ and recall the probability distribution
\begin{equation}
\sigma_{r,k,\ell}=\frac{1}{2}\sigma_{r,k,\ell}^{0}+\frac{1}{2}\sigma_{r,k,\ell}^{1},\label{eq:sigma-r-k-l-restated}
\end{equation}
where $\sigma_{r,k,\ell}^{0}$ and $\sigma_{r,k,\ell}^{1}$ (both
defined in Section~\ref{sec:discrepancy-of-DISJ}) are probability
distributions on $\UDISJ_{r,k}^{-1}(0)$ and $\UDISJ_{r,k}^{-1}(1)$,
respectively. Consider the following averaging operator $L,$ which
linearly sends real functions $\chi$ on $(\zoo^{r\times k})^{n}$
to real functions on $\zoon$ according to 
\begin{align*}
(L\chi)(z)=\Exp_{X_{1}\sim\sigma_{r,k,\ell}^{z_{1}}}\cdots\Exp_{X_{n}\sim\sigma_{r,k,\ell}^{z_{n}}}\;\chi(X_{1},\dots,X_{n}) &  & (z\in\zoon).
\end{align*}
When $\chi$ is an $\ell$-cylinder intersection,
\begin{align}
|\widehat{L\chi}(S)| & =\left|\Exp_{z\in\zoon}\;\Exp_{X_{1}\sim\sigma_{r,k,\ell}^{z_{1}}}\cdots\Exp_{X_{n}\sim\sigma_{r,k,\ell}^{z_{n}}}\left[\chi(X_{1},\dots,X_{n})\prod_{i\in S}(-1)^{z_{i}}\right]\right|\nonumber \\
 & =\left|\Exp_{X_{1},\dots,X_{n}\sim\sigma_{r,k,\ell}}\left[\chi(X_{1},\dots,X_{n})\prod_{i\in S}(-1)^{\UDISJ_{r,k}(X_{i})}\right]\right|\nonumber \\
 & \leq\disc_{\ell}\left(\DISJ_{r,k}^{\oplus|S|},(\sigma_{r,k,\ell})^{\otimes|S|}\right)\nonumber \\
 & \leq\left(\frac{2^{\ell}-1}{r}\cdot\left(\binom{k}{1}+\binom{k}{2}+\cdots+\binom{k}{\ell}\right)\right)^{|S|/2}\nonumber \\
 & \leq\left(\frac{1}{\sqrt{r}}\cdot\binom{k}{\mathord\leq\ell}\right)^{|S|},\label{eqn:coeff-bound}
\end{align}
where the second and fourth steps use (\ref{eq:sigma-r-k-l-restated})
and Theorem~\ref{thm:symmetric-disj-few-parties}, respectively.

Now, fix a randomized communication protocol for $F$ with error $\epsilon$
and cost $R_{\epsilon}(F).$ Approximate $F$ as in Corollary~\ref{cor:cylinders-function}
by a linear combination of $\ell$-cylinder intersections $\Pi=\sum_{\chi}a_{\chi}\chi,$
where $\sum_{\chi}|a_{\chi}|\leq2^{R_{\epsilon}(F)}.$ We claim that
$L\Pi$ is approximable by a low-degree polynomial. Indeed, let $d$
be a positive integer to be chosen later. Discarding the Fourier coefficients
of $L\Pi$ of order $d$ and higher gives 
\begin{align}
E(L\Pi,d-1)%
%|\widehat{L\Pi}(S)|
 & \leq\min\left\{ 1,\;\sum_{\chi}|a_{\chi}|\sum_{|S|\geq d}|\widehat{L\chi}(S)|\right\} \nonumber \\
 & \leq\min\left\{ 1,\;2^{R_{\epsilon}(F)}\sum_{i=d}^{n}{n \choose i}\left(\frac{1}{\sqrt{r}}\cdot\binom{k}{\mathord\leq\ell}\right)^{i}\right\} \nonumber \\
 & \leq\min\left\{ 1,\;2^{R_{\epsilon}(F)}\sum_{i=d}^{n}\left(\frac{\e n}{d\sqrt{r}}\cdot\binom{k}{\mathord\leq\ell}\right)^{i}\right\} \nonumber \\
 & \leq2^{R_{\epsilon}(F)}\left(\frac{2\e n}{d\sqrt{r}}\cdot\binom{k}{\mathord\leq\ell}\right)^{d},\label{eqn:MP-poly}
\end{align}
where the second and third steps use~(\ref{eqn:coeff-bound}) and~(\ref{eq:binom-sum-bound}),
respectively. On the other hand, recall from Corollary~\ref{cor:cylinders-function}
that $\Pi$ approximates $F$ in the sense that $\|\Pi\|_{\infty}\leq1$
and $|F-\Pi|\leq\epsilon$ on the domain of $F.$ It follows that
$\|L\Pi\|_{\infty}\leq1$ and $|f-L\Pi|\leq\epsilon$ on the domain
of $f,$ whence 
\begin{align*}
E(f,d-1)\leq\epsilon+E(L\Pi,d-1).
\end{align*}
Substituting the estimate from (\ref{eqn:MP-poly}), 
\begin{align*}
E(f,d-1) & \leq\epsilon+2^{R_{\epsilon}(F)}\left(\frac{2\e n}{d\sqrt{r}}\cdot\binom{k}{\mathord\leq\ell}\right)^{d}.
\end{align*}
For $d=\deg_{\delta}(f),$ the left-hand side of this inequality by
definition exceeds $\delta,$ completing the proof. 
\end{proof}
We now specialize the previous theorem to set disjointness.
\begin{thm}
\label{thm:disj-cc-lower}Abbreviate $R=R_{1/4}(\UDISJ_{n,k}).$ Then
\begin{equation}
\binom{k}{\mathord\leq R}^{2}R^{4}\geq\Omega(n).\label{eq:disj-master-formula}
\end{equation}
In particular,
\begin{align}
R_{1/3}(\DISJ_{n,k}) & \geq R_{1/3}(\UDISJ)=\Omega\left(\frac{\log n}{\log\left\lceil 1+\frac{k}{\log n}\right\rceil }+1\right).\label{eq:disj-cc-PS}
\end{align}
\end{thm}

\begin{proof}
Observe that $R_{1/4}(\UDISJ_{n,k})$ is monotonically nondecreasing
in $n.$ Now for all $1\leq r\leq n,$ we have $R\geq R_{1/4}(\UAND_{\lfloor n/r\rfloor}\circ\UDISJ_{r,k})$
by (\ref{eq:udisj-decompose}) and therefore 
\[
2^{R}\geq\left(\frac{1}{3}-\frac{1}{4}\right)\left(\frac{\deg_{1/3}(\UAND_{\lfloor n/r\rfloor})\sqrt{r}}{2\e\lfloor n/r\rfloor\binom{k}{\leq R}}\right)^{\deg_{1/3}(\UAND_{\lfloor n/r\rfloor})}
\]
by Theorem~\ref{thm:randomized-VP}. This in turn simplifies to 
\begin{equation}
2^{R}\geq\frac{1}{12}\left(\frac{cr}{\sqrt{n}\binom{k}{\leq R}}\right)^{c\sqrt{n/r}}\label{eq:intermediate-master-disj}
\end{equation}
for some absolute constant $c>0,$ by Theorem~\ref{thm:nisan-szegedy}.
There are two cases to examine. If $\binom{k}{\leq R}\geq\frac{c}{2}\sqrt{n}$,
then (\ref{eq:disj-master-formula}) is trivially true. Otherwise,
letting $r=\lceil\frac{2}{c}\sqrt{n}\binom{k}{\leq R}\rceil$ in (\ref{eq:intermediate-master-disj})
forces (\ref{eq:disj-master-formula}).

It remains to explain how the newly obtained relation (\ref{eq:disj-master-formula})
implies the communication lower bounds in the theorem statement. By
(\ref{eq:binom-sum-bound}), 
\[
\binom{k}{\mathord\leq R}\leq\e^{R}\left\lceil \frac{k}{R}\right\rceil ^{R}.
\]
Substituting this estimate in (\ref{eq:disj-master-formula}) gives
$\e^{2R}\lceil k/R\rceil^{2R}R^{4}\geq\Omega(n),$ whence (\ref{eq:disj-cc-PS}).
\end{proof}
\begin{rem}
Theorem~\ref{thm:randomized-VP} is sufficiently general to imply
our lower bound for generalized inner product as well (Theorem~\ref{thm:gip-cc-lower}).
However, considering the effort required to prove Theorem~\ref{thm:randomized-VP}
itself, the treatment of $\GIP_{n,k}$ in Section~\ref{sec:gip-cc-lower}
appears to be more direct.
\end{rem}

\subsection{Upper bound}

To prove a matching upper bound on the communication complexity of
set disjointness for $k\geq\log n$ players, we reduce this problem
to inner product and appeal to our previously established Theorem~\ref{thm:gip-upper-cc}.
\begin{thm}
\label{thm:disj-upper-cc}For any $k\geq\log n$ and any constant
$\epsilon>0,$
\begin{align*}
R_{\epsilon}(\text{\DISJ}_{n,k}) & =O\left(\frac{\log n}{\log\left\lceil 1+\frac{k}{\log n}\right\rceil }+1\right).
\end{align*}
Moreover, this upper bound is achieved by a simultaneous protocol.
\end{thm}

\begin{proof}
Recall that for any string $y\in\zoon,$ 
\[
\Prob_{S\subseteq\{1,2,\dots,n\}}\left[\bigoplus_{i\in S}y_{i}=0\right]=\begin{cases}
1 & \text{if \ensuremath{y=0^{n},}}\\
1/2 & \text{otherwise.}
\end{cases}
\]
This gives the following well-known relation between generalized inner
product and set disjointness: 
\[
\Prob_{S\subseteq\{1,2,\dots,n\}}\left[\GIP_{|S|,k}(X|_{S})=0\right]=\begin{cases}
1 & \text{if \ensuremath{\DISJ_{n,k}(X)=1,}}\\
1/2 & \text{otherwise.}
\end{cases}
\]
Thus, the players can solve an instance $X\in\zoo^{n\times k}$ of
set disjointness by estimating $\Prob_{S}[\GIP_{|S|,k}(X|_{S})=0]$.
This can be done by running the protocol of Theorem~\ref{thm:gip-upper-cc}
with error parameter $1/4$ a constant number of times, simultaneously
and in parallel, each time on a uniformly random subset of the rows
of $X$.
\end{proof}

\section{Communication bounds independent of $k$\label{sec:indep-of-k}}

In this final section, we study the communication problem $F_{n,k}\colon(\zoon)^{k}\to\zoo$
given by
\begin{equation}
F_{n,k}(X)=\MOD_{3}\left(\bigoplus_{j=1}^{k}X_{1,j},\ldots,\bigoplus_{j=1}^{k}X_{n,j}\right).\label{eq:MOD3-of-XORs-definition}
\end{equation}
Ada et al.~\cite{ACFN12nof-composed-functions} proved that this
function has randomized communication complexity $R_{1/3}(F_{n,k})=\Omega(n/4^{k}).$
Here, we derive the incomparable lower bound $R_{1/3}(F_{n,k})\geq\frac{1}{3}\log n-\frac{1}{3}$
for all $n$ and $k,$ which shows that $F_{n,k}$ requires nontrivial
communication regardless of the number of players $k.$ We also prove
that for every $k\geq\log n,$ our lower bound is tight up to a multiplicative
constant.

\subsection{Lower bound}

In what follows, we use the shorthand $e(t)=\exp(2\pi\iu t),$ where
$\iu$ is the imaginary unit. Our proof requires the following correlation
bound for cylinder intersections, which is implied by the more general
work of Ada et al.~\cite{ACFN12nof-composed-functions}.
\begin{lem}[cf.~Ada et al.]
\label{lem:exp-correl-k}For every cylinder intersection $\chi\colon(\zoon)^{k}\to\zoo,$
\begin{align*}
\left|\Exp_{X\in\zoo^{n\times k}}\,e\left(\frac{1}{3}\sum_{i=1}^{n}\bigoplus_{j=1}^{k}X_{i,j}\right)\chi(X)\right| & <\exp\left(-\frac{n}{4^{k}}\right).
\end{align*}
\end{lem}

\noindent For the reader's convenience, a short and self-contained
proof of Lemma~\ref{lem:exp-correl-k} as stated above is available
in Appendix~\ref{app:exp-correl}. We now sharpen this result for
$\ell$-cylinder intersections with $\ell\leq k.$
\begin{lem}
\label{lem:exp-correl-ell}For every $\ell$-cylinder intersection
$\chi\colon(\zoon)^{k}\to\zoo,$
\begin{align}
\left|\Exp_{X\in\zoo^{n\times k}}\;\;e\!\left(\frac{1}{3}\sum_{i=1}^{n}\bigoplus_{j=1}^{k}X_{i,j}\right)\chi(X)\right| & <\exp\left(-\frac{n}{4^{\ell}}\right).\label{eq:exp-correl-ell}
\end{align}
\end{lem}

\begin{proof}
By symmetry, we may assume that $\chi$ is a $\{1,2,\ldots,\ell\}$-cylinder
intersection. In what follows, we let $X$ stand for a uniformly random
matrix in $\zoo^{n\times k}.$ Define auxiliary random variables $X',X''$
by
\begin{align*}
X' & =\begin{bmatrix}X_{1,1} & X_{1,2} & \cdots & X_{1,\ell-1} & X_{1,\ell}\oplus\cdots\oplus X_{1,k}\\
X_{2,1} & X_{2,2} & \cdots & X_{2,\ell-1} & X_{2,\ell}\oplus\cdots\oplus X_{2,k}\\
\vdots & \vdots &  & \vdots & \vdots\\
X_{n,1} & X_{n,2} & \cdots & X_{n,\ell-1} & X_{n,\ell}\oplus\cdots\oplus X_{n,k}
\end{bmatrix},\\
X'' & =\begin{bmatrix}X_{1,\ell+1} & X_{1,\ell+2} & \cdots & X_{1,k}\\
X_{2,\ell+1} & X_{2,\ell+2} & \cdots & X_{2,k}\\
\vdots & \vdots &  & \vdots\\
X_{n,\ell+1} & X_{n,\ell+2} & \cdots & X_{n,k}
\end{bmatrix}.
\end{align*}
In this notation,
\begin{multline}
\Exp\left[e\left(\frac{1}{3}\sum_{i=1}^{n}\bigoplus_{j=1}^{k}X_{i,j}\right)\chi(X)\;\middle|\;X''\right]\\
=\Exp\left[e\left(\frac{1}{3}\sum_{i=1}^{n}\bigoplus_{j=1}^{\ell}X_{i,j}'\right)\chi(X)\;\middle|\;X''\right].\qquad\qquad\label{eq:conditional-X2}
\end{multline}
Setting $X''$ to any given value makes $\chi(X)$ a cylinder intersection
in terms of $X'.$ Moreover, the conditional probability distribution
of $X'$ given $X''$ is uniform on $\zoo^{n\times\ell}$. In view
of Fact~\ref{lem:exp-correl-k}, we conclude that the expectation
on the right-hand side of~(\ref{eq:conditional-X2}) is smaller in
absolute value than $\exp(-n/4^{\ell}).$ Averaging over $X''$ now
gives~(\ref{eq:exp-correl-ell}).
\end{proof}
We are now in a position to prove our claimed lower bound on the communication
complexity of $F_{n,k}.$
\begin{thm}
Let $F_{n,k}\colon(\zoon)^{k}\to\zoo$ be given by~\emph{(\ref{eq:MOD3-of-XORs-definition})}.
Then
\begin{equation}
\disc_{\ell}(F_{n,k},\nu)\leq2\exp\left(-\frac{n}{4^{\ell}}\right),\label{eq:MOD3-of-XOR-disc-bound}
\end{equation}
where $\nu$ is the probability distribution under which the weight
of any point of $F_{n,k}^{-1}(1)$ is double the weight of any point
of $F_{n,k}^{-1}(0).$ In particular,
\begin{equation}
R_{1/3}(F_{n,k})\geq\frac{1}{3}\log n-\frac{1}{3}.\label{eq:MOD3-of-XOR-cc-bound}
\end{equation}
\end{thm}

\begin{proof}
Observe that
\[
e\left(\frac{t}{3}\right)+\overline{e\left(\frac{t}{3}\right)}=\begin{cases}
2 & \text{if }t\equiv0\pmod3,\\
-1 & \text{if }t\equiv\pm1\pmod3,
\end{cases}
\]
where the bar denotes complex conjugation. As a result,
\[
e\!\left(\frac{1}{3}\sum_{i=1}^{n}\bigoplus_{j=1}^{k}X_{i,j}\right)+\overline{e\!\left(\frac{1}{3}\sum_{i=1}^{n}\bigoplus_{j=1}^{k}X_{i,j}\right)}\equiv c\cdot2^{nk}(-1)^{F_{n,k}(X)+1}\nu(X)
\]
for some normalizing constant $c\geq1.$ We conclude that for every
$\ell$-cylinder intersection $\chi,$
\begin{align*}
\left|\Exp_{X\sim\nu}(-1)^{F_{n,k}(X)}\chi(X)\right| & =\frac{1}{c}\left|\Exp_{X\in\zoo^{n\times k}}\left[e\!\left(\frac{1}{3}\sum_{i=1}^{n}\bigoplus_{j=1}^{k}X_{i,j}\right)\chi(X)\right.\right.\\
 & \qquad\qquad\qquad\qquad\left.\left.+\overline{e\!\left(\frac{1}{3}\sum_{i=1}^{n}\bigoplus_{j=1}^{k}X_{i,j}\right)}\chi(X)\right]\right|\\
 & \leq2\left|\Exp_{X\in\zoo^{n\times k}}\left[e\!\left(\frac{1}{3}\sum_{i=1}^{n}\bigoplus_{j=1}^{k}X_{i,j}\right)\chi(X)\right]\right|\\
 & \leq2\exp\left(-\frac{n}{4^{\ell}}\right),
\end{align*}
where the final step uses Lemma~\ref{lem:exp-correl-ell}. This establishes~(\ref{eq:MOD3-of-XOR-disc-bound}).
Now Theorem~\ref{thm:dm} yields
\[
2^{R_{1/3}(F_{n,k})}\geq\frac{1}{6}\exp\left(\frac{n}{4^{R_{1/3}(F_{n,k})}}\right),
\]
which implies~(\ref{eq:MOD3-of-XOR-cc-bound}) by elementary calculus.
\end{proof}

\subsection{Upper bound}

We now show that our lower bound on the communication complexity of
$F_{n,k}$ is tight up to a constant factor for all $k\geq\log n.$
The idea is to alter $F_{n,k}$ in a random fashion on a small portion
of the domain so as to make it representable by a polynomial of degree
$k-1,$ which the players can then directly evaluate. This technique
was previously used in~\cite{ACFN12nof-composed-functions,chattopadhyay-saks14superlog-nof}.
\begin{thm}
\label{thm:mod-xor-upper-cc}Let $F_{n,k}\colon(\zoo^{n})^{k}\to\zoo$
be given by~\emph{(\ref{eq:MOD3-of-XORs-definition}).} Then
\[
R_{1/3}(F_{n,k})\leq c\log n+c
\]
for some absolute constant $c>0$ and all $n\geq1$ and $k\geq\log n.$
Moreover, this upper bound is achieved by a simultaneous protocol.
\end{thm}

\begin{proof}
For $u\in\zook,$ consider the multivariate polynomial $p_{u}\in\ZZ_{3}[x_{1},x_{2},\ldots,x_{k}]$
given by
\[
p_{u}(x_{1},x_{2},\ldots,x_{k})=\prod_{i=1}^{k}(x_{i}+1)-\prod_{i=1}^{k}(x_{i}+u_{i}-1)-1.
\]
Then
\begin{align}
 & \deg p_{u}\leq k-1,\label{eq:p_u-deg}\\
 & p_{u}(x_{1},x_{2},\ldots,x_{k})=x_{1}\oplus x_{2}\oplus\cdots\oplus x_{k}, &  & x\in\zook\setminus\{u\}.\label{eq:p_u-XOR}
\end{align}
Our protocol for $F_{n,k}$ is as follows. On input $X\in\zoo^{n\times k},$
the players pick a random point $u\in\zook$ using shared randomness,
and consider the following polynomial $F_{n,k,u}\in\ZZ_{3}[X_{1,1},\ldots,X_{n,k}]$:
\[
F_{n,k,u}(X)=\sum_{i=1}^{n}p_{u}(X_{i,1},X_{i,2},\ldots,X_{i,k}).
\]
It follows from~(\ref{eq:p_u-deg}) that $\deg F_{n,k,u}\leq k-1,$
which makes it possible to partition the monomials of $F_{n,k,u}$
among the players in some predetermined fashion and have each player
report the sum of the monomials assigned to him. This simultaneous
protocol has cost $k\lceil\log3\rceil=2k.$ With probability at least
$1-n2^{-k}$ over shared randomness, we have $u\ne(X_{i,1},X_{i,2},\ldots,X_{i,k})$
for $i=1,2,\ldots,n,$ in which case~(\ref{eq:p_u-XOR}) implies
that $F_{n,k}(X)\equiv1-F_{n,k,u}(X)^{2}\pmod3.$ In particular, the
players' broadcasts uniquely identify $F_{n,k}(X)$ with probability
at least $1-n2^{-k}$ on any given input $X.$ Thus,
\begin{equation}
R_{n2^{-k}}(F_{n,k})\leq2k.\label{eq:mod-xor-protocol}
\end{equation}
To complete the proof of the theorem, we consider three cases depending
on the value of $k.$
\begin{itemize}[topsep=3mm,itemsep=3mm]
\item Equation~(\ref{eq:mod-xor-protocol}) directly implies that $R_{1/3}(F_{n,\lceil\log3n\rceil})\leq2\lceil\log3n\rceil,$
which settles the case $k=\lceil\log3n\rceil.$ As usual, the error
probability can be reduced from $1/3$ to any positive constant by
running the protocol simultaneously multiple times and outputting
the majority answer. In particular, $R_{\epsilon}(F_{n,\lceil\log3n\rceil})=O(\log n)$
for any fixed $\epsilon>0.$
\item For $k>\lceil\log3n\rceil$, observe that
\begin{multline*}
F_{n,k}(X)=\MOD_{3}\left(X_{1,1}\oplus\cdots\oplus X_{1,\lceil\log3n\rceil-1}\oplus\left(\bigoplus_{j=\lceil\log3n\rceil}^{k}X_{1,j}\right),\ldots\right.\\
\left.X_{n,1}\oplus\cdots\oplus X_{n,\lceil\log3n\rceil-1}\oplus\left(\bigoplus_{j=\lceil\log3n\rceil}^{k}X_{n,j}\right)\right)
\end{multline*}
and in particular $R_{\epsilon}(F_{n,k})\leq R_{\epsilon}(F_{\lceil\log3n\rceil,k})$
for all $\epsilon.$ As a result, we conclude by the first case that
$R_{\epsilon}(F_{n,k})\leq O(\log n)$ for any fixed $\epsilon>0$
and $k>\lceil\log3n\rceil.$
\item For $\log n\leq k<\lceil\log3n\rceil,$ the players partition the
input matrix $X$ horizontally into submatrices with at most $n/3$
rows each and run the protocol from the previous cases with the error
parameter set to a small constant. Then with probability at least~$2/3,$
the players' broadcasts uniquely determine $\sum_{i=1}^{n}\bigoplus_{j=1}^{k}X_{i,j}$~modulo~$3$
and thereby reveal $F_{n,k}(X).\qedhere$
\end{itemize}
\end{proof}
\bibliographystyle{siamplain}
\bibliography{venues,sherstov}

\appendix

\section{A correlation bound\label{app:exp-correl}}

The purpose of this appendix is to provide a short proof of Lemma~\ref{lem:exp-correl-k}.
The only technical prerequisite for the proof is the following fact
about cylinder intersections, which is implicit in the work of Babai
et al.~\cite{bns92} and is derived explicitly in the followup papers
by Chung and Tetali~\cite{chung-tetali93} and Raz~\cite{raz00multiparty}.
\begin{fact}[Chung and Tetali; Raz]
\label{fact:bns-bound-on-discrepancy}Let $U_{1},U_{2},\ldots,U_{k}$
be finite sets. Then for any function $\phi\colon U_{1}\times U_{2}\times\cdots\times U_{k}\to\mathbb{C}$
and any cylinder intersection $\chi\colon U_{1}\times U_{2}\times\cdots\times U_{k}\to\zoo,$
\begin{multline*}
\left|\Exp_{u_{1},\ldots,u_{k}}\phi(u_{1},\ldots,u_{k})\chi(u_{1},\ldots,u_{k})\right|^{2^{k}}\\
\leq\Exp_{\substack{u_{1}^{0},\ldots,u_{k}^{0}\\
u_{1}^{1},\ldots,u_{k}^{1}
}
}\left[\prod_{\substack{z\in\zook\\
\text{\ensuremath{|z|} even}
}
}\phi(u_{1}^{z_{1}},\ldots,u_{k}^{z_{k}})\cdot\prod_{\substack{z\in\zook\\
\text{\ensuremath{|z|} odd}
}
}\overline{\phi(u_{1}^{z_{1}},\ldots,u_{k}^{z_{k}})}\right],
\end{multline*}
where $u_{i},u_{i}^{0},u_{i}^{1}\in U_{i}$ for each $i.$
\end{fact}

We are now in a position to prove Lemma~\ref{lem:exp-correl-k},
which we restate below for the reader's convenience. Recall that we
use the shorthand $e(t)=\exp(2\pi\iu t),$ where $\iu$ is the imaginary
unit. 
\begin{lem*}[cf. Ada et al.]
For every cylinder intersection $\chi\colon(\zoon)^{k}\to\zoo,$
\begin{align*}
\left|\Exp_{X\in\zoo^{n\times k}}\,e\left(\frac{1}{3}\sum_{i=1}^{n}\bigoplus_{j=1}^{k}X_{i,j}\right)\chi(X)\right| & <\exp\left(-\frac{n}{4^{k}}\right).
\end{align*}
\end{lem*}
\begin{proof}
We mostly follow the analysis of Ada et al.~\cite{ACFN12nof-composed-functions},
who proved a more general correlation bound. As one would expect,
focusing on a special case as we do here allows for a shorter and
simpler presentation. We have:
\begin{align}
 & \left|\Exp_{X\in\zoo^{n\times k}}\,e\left(\frac{1}{3}\sum_{i=1}^{n}\bigoplus_{j=1}^{k}X_{i,j}\right)\chi(X)\right|^{2^{k}}\nonumber \\
 & \qquad\qquad\leq\Exp_{X^{0},X^{1}\in\zoo^{n\times k}}\;\prod_{z\in\zook}e\left(\frac{(-1)^{|z|}}{3}\sum_{i=1}^{n}\bigoplus_{j=1}^{k}X_{i,j}^{z_{j}}\right)\nonumber \\
 & \qquad\qquad=\Exp_{X^{0},X^{1}\in\zoo^{n\times k}}\;\prod_{i=1}^{n}e\left(\frac{1}{3}\sum_{z\in\zook}(-1)^{|z|}\bigoplus_{j=1}^{k}X_{i,j}^{z_{j}}\right)\nonumber \\
 & \qquad\qquad=\left(\Exp_{x^{0},x^{1}\in\zook}\;e\left(\frac{1}{3}\sum_{z\in\zook}(-1)^{|z|}\bigoplus_{j=1}^{k}x_{j}^{z_{j}}\right)\right)^{n}\nonumber \\
 & \qquad\qquad=\left(\Exp_{x^{0},x^{1}\in\zook}\;e\left(\frac{1}{3}\sum_{z\in\zook}(-1)^{|z|}\left(\prod_{j=1}^{k}(x_{j}^{z_{j}}+1)-1\right)\right)\right)^{n}\nonumber \\
 & \qquad\qquad=\left(\Exp_{x^{0},x^{1}\in\zook}\;e\left(\frac{1}{3}\prod_{j=1}^{k}(x_{j}^{0}-x_{j}^{1})\right)\right)^{n},\label{eq:disc-mod-halfway}
\end{align}
where the first step follows from Fact~\ref{fact:bns-bound-on-discrepancy},
and the fourth step uses the fact that 
\[
\bigoplus_{j=1}^{k}x_{j}^{z_{j}}\equiv\prod_{j=1}^{k}(x_{j}^{z_{j}}+1)-1\pmod3.
\]
 A routine calculation reveals that
\[
e\left(\frac{1}{3}\prod_{j=1}^{k}(x_{j}^{0}-x_{j}^{1})\right)=\begin{cases}
1 & \text{if \ensuremath{x_{j}^{0}=x_{j}^{1}} for some \ensuremath{j},}\\
-\frac{1}{2}+\frac{\sqrt{3}}{2}\iu\prod_{j=1}^{k}(x_{j}^{0}-x_{j}^{1}) & \text{otherwise,}
\end{cases}
\]
where $\iu$ is the imaginary unit. Making this substitution in~(\ref{eq:disc-mod-halfway})
yields
\begin{align*}
\left|\Exp_{X\in\zoo^{n\times k}}\,e\left(\frac{1}{3}\sum_{i=1}^{n}\bigoplus_{j=1}^{k}X_{i,j}\right)\chi(X)\right|^{2^{k}} & \leq\left(\left(1-\frac{1}{2^{k}}\right)\cdot1-\frac{1}{2^{k}}\cdot\frac{1}{2}\right)^{n}\\
 & \leq\left(1-\frac{1}{2^{k}}\right)^{n}\\
 & <\exp\left(-\frac{n}{2^{k}}\right).\qedhere
\end{align*}
\end{proof}

\end{document}